\documentclass[12pt]{amsart}
\usepackage{}

\usepackage{amsmath}
\usepackage{amsfonts}
\usepackage{amssymb}
\usepackage[all]{xy}           
\usepackage{bbm}
\usepackage{bbding}
\usepackage{txfonts}
\usepackage{amscd}
\usepackage{tikz}
\usetikzlibrary{matrix}
\usepackage[shortlabels]{enumitem}

\usepackage{ifpdf}
\ifpdf
\usepackage[colorlinks,final,backref=page,hyperindex]{hyperref}
\else
\usepackage[colorlinks,final,backref=page,hyperindex,hypertex]{hyperref}
\fi
\usepackage{tikz}
\usepackage[active]{srcltx}

\makeatletter

\newtheorem{thm}{Theorem}[section]
\newtheorem{lem}[thm]{Lemma}

\newtheorem{pro}[thm]{Proposition}
\newtheorem{ex}[thm]{Example}

\newtheorem{defi}[thm]{Definition}

\newcommand {\emptycomment}[1]{}
\newcommand {\yh}[1]{{\marginpar{*}\scriptsize\textcolor{purple}{yh: #1}}}

\newcommand{\lon }{\,\rightarrow\,}
\newcommand{\be }{\begin{equation}}
\newcommand{\ee }{\end{equation}}

\newcommand{\g}{\mathfrak g}
\newcommand{\h}{\mathfrak h}

\newcommand{\huaC}{{\mathfrak{C}}}

\newcommand{\huaH}{\mathcal{H}}

\newcommand{\frkg}{\mathfrak g}

\newcommand{\frkX}{\mathfrak X}

\newcommand{\half}{\frac{1}{2}}

\newcommand{\Courant}[1]{\left\llbracket  #1\right\rrbracket }

\newcommand{\jetd}{\mathbbm{d}}

\newcommand{\Id}{\rm{Id}}

\newcommand{\p}{\mathbbm{p}}
\newcommand{\br}[1]{   [ \cdot,    \cdot  ]   }
\newcommand{\id}{\mathbbm{i}}

\newcommand{\dM}{\mathrm{d}}

\newcommand{\Hom}{\mathrm{Hom}}

\newcommand{\Nat}{\mathbb N}
\newcommand{\Der}{\mathrm{Der}}

\newcommand{\IDer}{\mathrm{InnDer}}

\newcommand{\gl}{\mathfrak {gl}}

\newcommand{\AD}{\mathfrak{ad}}

\newcommand{\ad}{\mathrm{ad}}

\newcommand{\Img}{\mathrm{Im}}

\newcommand{\sgn}{\mathrm{sgn}}

\newcommand{\K}{\mathbb{K}}

\newcommand{\perm}{\mathbb S}

\newcommand{\NR}{\mathrm{NR}}

\begin{document}

\title[Maurer-Cartan characterizations and cohomologies of compatible Lie
algebras]{Maurer-Cartan characterizations and cohomologies of compatible Lie
algebras}

\author{Jiefeng Liu}
\address{School of Mathematics and Statistics, Northeast Normal University, Changchun 130024, China}
\email{liujf534@nenu.edu.cn}
\author{Yunhe Sheng}
\address{Department of Mathematics, Jilin University, Changchun 130012, Jilin, China}
\email{shengyh@jlu.edu.cn}
\author{Chengming Bai}
\address{Chern Institute of Mathematics \& LPMC, Nankai University, Tianjin 300071, China}
\email{baicm@nankai.edu.cn}


\begin{abstract}
  In this paper, we give Maurer-Cartan characterizations as
well as a cohomology theory for compatible Lie algebras.
Explicitly, we first introduce the notion of a bidifferential
graded Lie algebra and  thus give
Maurer-Cartan characterizations of compatible Lie algebras. Then
we introduce a cohomology theory of compatible Lie algebras and
use it to   classify infinitesimal
deformations and abelian extensions of compatible Lie algebras. In
particular, we introduce the reduced cohomology of a compatible
Lie algebra and establish the relation between the reduced
cohomology of a compatible Lie algebra and the cohomology of the
corresponding compatible linear Poisson structures introduced by
Dubrovin and Zhang in their study of bi-Hamiltonian structures.
Finally, we use the Maurer-Cartan approach to classify nonabelian
extensions of compatible Lie algebras.
\end{abstract}



\maketitle

\tableofcontents

\allowdisplaybreaks


\section{Introduction}\label{sec:intr}


  A compatible Lie algebra is a pair of Lie algebras such
that any linear combination of these two Lie brackets is still a
Lie bracket. Such structures closely relate to the
(infinitesimal) deformations of Lie algebras. Compatible Lie
algebras appear in a lot of fields in mathematics and mathematical
physics. For example, compatible Lie algebras   naturally
one-to-one correspond  to compatible linear Poisson structures,
whereas the latter are important Poisson structures related to
bi-Hamiltonian structures. Recall that a bi-Hamiltonian structure
is a pair of Poisson structures on a manifold which are
compatible, i.e. their sum is again a Poisson structure. Note that
the pioneering work in the integrable bi-Hamiltonian systems was
done by Magri in \cite{Mag} and many completely integrable
Hamiltonian systems arising in mechanics, mathematical physics and
geometry have the bi-Hamiltonian structures
(\cite{GD,Kosmann1,MaMo}). Bi-Hamiltonian structures are very
powerful to construct separation of variables and description of
properties of solutions in the theory of integrable Hamiltonian
systems. See the book \cite{LPV} for more details on linear
Poisson structures and \cite{Bol1,Bol2,Rey,Tsy} for more details
of compatible Lie algebras and  compatible linear Poisson
structures and their applications. Furthermore, compatible Lie
algebras also appear in the study of the classical Yang-Baxter
equation and principal chiral field (\cite{GS1}), loop algebras
over Lie algebras (\cite{GS2}) and  elliptic theta functions
(\cite{GS3}). Also see \cite {DK,Pan,S,WB} for more details on
operads, classifications  and bialgebra theory of compatible Lie
algebras.

The cohomology theory is a classical approach  associating
invariants to a mathematical structure.  Cohomology controls
deformations and extension problems of the corresponding algebraic
structures. Cohomology theories of various kinds of algebras have
been developed and studied in \cite{Ch-Ei,Ge0,Har,Hor}. See the
review paper \cite{GLST} for more details.

In this paper, we study the cohomology theory of compatible Lie
algebras. By using the bidifferential graded Lie algebra
constructed by a compatible Lie algebra,   we give a
cohomology theory of compatible Lie algebras.   We would like to
point out that the cohomology theory of compatible Lie algebras
given here is not merely a combination of cohomologies of the two
Lie algebras as many other properties or structures of compatible
Lie algebras do so. In particular, we construct a new cochain
complex for a compatible Lie algebra  which is independent of the
usual cochain complex for the Chevalley-Eilenberg cohomology of
a Lie algebra.
Then we study infinitesimal deformations of a compatible Lie
algebra. We show that infinitesimal deformations of a compatible
Lie algebra  are classified by the second cohomology group.   In particular, it leads to the
introduction of the notion of a Nijenhuis operator on a compatible
Lie algebra which generates a trivial deformation of the
compatible Lie algebra, which is similar to the study of Nijenhuis
operators on  Lie algebras that play an important role in
the study of integrability of nonlinear evolution equations
(\cite{Dorf}).

Furthermore, we introduce the cohomology of a compatible Lie
algebra  associated to arbitrary representation.   We would also like to
point out that there might be more than one cohomology  theory  for compatible Lie algebras. Such a phenomenon
has already appeared in the study of cohomologies of some
algebraic systems with more than one product  such as Poisson
algebras (\cite{CohomologyPA1,FGV}). In particular, the coproduct
for a compatible Lie bialgebra given in \cite{WB} is not a
$1$-cocycle in the sense of the cohomology theory in this paper,
that is, there might be another cohomology theory for compatible
Lie algebras such that it is a $1$-cocycle. Then we study abelian
extensions of a compatible Lie algebra  and show that they are
classified by the second cohomology group.   We also
introduce the reduced cohomology of a compatible Lie algebra and
establish the relation between the reduced cohomology of a
compatible Lie algebra and the cohomology of the corresponding
compatible linear Poisson structures introduced by Dubrovin and
Zhang in \cite{DZ} in the study of bi-Hamiltonian structures.
Finally, we study nonabelian extensions of compatible Lie algebras
by a similar method introduced in \cite{nonabelin cohomology of
Lie}. We construct a bidifferential graded Lie algebra associated
to compatible Lie algebras $\g$ and $\h$ and describe nonabelian
extensions of $\g$ by $\h$ by its Maurer-Cartan elements.

The paper is organized as follows. In Section \ref{sec:NR}, we
briefly recall the cohomology theory and the Nijenhuis-Richardson
bracket for Lie algebras. In Section \ref{sec:cohomology I}, we
first recall the notion of compatible Lie algebras and their
representations. Then we introduce the notions of a bidifferential
graded  Lie algebra and its Maurer-Cartan elements. In particular,
we give the Maurer-Cartan characterizations of compatible Lie
algebras. Next we introduce
a cohomology theory  of compatible Lie algebras and study
infinitesimal deformations of compatible Lie algebras. In Section
\ref{sec:cohomology II}, we introduce a cohomology theory of a compatible
Lie algebra  with coefficients in a representation and use this
cohomology to classify abelian extensions of a compatible Lie
algebra.   Then we introduce the reduced cohomology of a
compatible Lie algebra and establish the relation between the
reduced cohomology of a compatible Lie algebra and the cohomology
of the corresponding compatible linear Poisson structures. In
Section \ref{sec:cohomology III}, we use the Maurer-Cartan
approach to classify nonabelian extensions of compatible Lie
algebras.

In this paper, all the vector spaces are over algebraically closed field $\mathbb K$ of characteristic $0$, and finite dimensional.
\vspace{2mm}

{\bf Acknowledgements. }This research was  supported by NSFC (11901501,11922110,11931009). C. Bai is also supported by the Fundamental Research Funds for the Central Universities and Nankai ZhiDe
 Foundation.

\section{The Nijenhuis-Richardson bracket}\label{sec:NR}

 In this section, we recall the cohomology theory for Lie algebras
and the Nijenhuis-Richardson bracket.

A permutation $\sigma\in\perm_n$ is called an $(i,n-i)$-unshuffle if $\sigma(1)<\cdots<\sigma(i)$ and $\sigma(i+1)<\cdots<\sigma(n)$. If $i=0$ or $i=n$, we assume $\sigma=\Id$. The set of all $(i,n-i)$-unshuffles will be denoted by $\perm_{(i,n-i)}$.

A {\bf Maurer-Cartan element} of a differential graded Lie algebra $(L=\oplus_{i\in\mathbb Z} L_i,[-,-],\partial)$ is an element $P\in L_1$ such that
$$\partial P+\frac{1}{2}[P,P]=0.$$

Let $\g$ be a vector space. We consider the graded vector space $C^*(\g,\g)=\oplus_{n=0}^{+\infty}C^n(\g,\g)$, where $C^n(\g,\g)=\Hom(\wedge^{n}\g,\g)$  and the degree of   elements in $C^n(\g,\g)$ are defined to be $n-1$. Then $C^*(\g,\g)$ equipped with the {\bf Nijenhuis-Richardson bracket}
\begin{equation}
  [P,Q]_{\NR}=P\circ Q-(-1)^{pq}Q\circ P,\quad \forall~P\in C^{p+1}(\g,\g),Q\in C^{q+1}(\g,\g)
\end{equation}
is a graded Lie algebra, where $P\circ Q\in C^{p+q+1}(\g,\g)$ is defined by
\begin{equation}
  P\circ Q(x_1,\cdots,x_{p+q+1})=\sum_{\sigma\in \perm_{(q+1,p)}}(-1)^\sigma P(Q(x_{\sigma(1)},\cdots,x_{\sigma(q+1)}),x_{\sigma(q+2)},\cdots,x_{\sigma(p+q+1)}).
\end{equation}

 Note that any graded Lie algebra is a differential
graded Lie algebra with the zero differential.
 The following result is well known.
\begin{pro}{\rm(\cite{NR})}\label{pro:MCLie} Let $\g$ be a vector
space. Then $\pi\in\Hom(\wedge^2\g,\g)$ defines a Lie algebra if and only if $\pi$ is a Maurer-Cartan element of the graded Lie algebra $(C^*(\g,\g),[-,-]_{\NR})$, i.e. $[\pi,\pi]_{\NR}=0.$
\end{pro}

Let $(\g,\pi)$ be a Lie algebra. Because of the graded Jacobi identity, we get a coboundary operator $\dM^n_\pi:C^n(\g,\g)\rightarrow C^{n+1}(\g,\g)$ defined by
$$\dM^n_\pi\omega =(-1)^{n-1}[\pi,\omega]_\NR,\quad \forall~\omega\in C^n(\g,\g).$$
More precisely, for $\omega\in C^n(\g,\g)$ and $x_1,x_2,\cdots, x_{n+1}\in\g$, we have
\begin{eqnarray*}
    \dM^n_\pi\omega(x_1,x_2,\cdots,x_{n+1})&=&\sum_{i=1}^{n+1}(-1)^{i+1}\pi(x_i,\omega(x_1,\cdots,\hat{x_i},\cdots,x_{n+1}))\\
    &&+\sum_{1\leq i<j\leq n+1}(-1)^{i+j}\omega(\pi(x_i,x_j),x_1,\cdots,\hat{x_i},\cdots,\hat{x_j},\cdots,x_{n+1}),
\end{eqnarray*}
which is just the Chevalley-Eilenberg  coboundary operator for the Lie algebra $(\g, \pi)$ associated to the adjoint representation $(\g;\ad)$.

Let $\g_1$ and $\g_2$ be two vector spaces. The elements in $\g_1$ are denoted by $x,y, x_i$ and the elements in $\g_2$ are denoted by $u,v,v_i$. For a given linear map $f:\wedge^{k}\g_1\otimes\wedge^{l}\g_2\lon\g_1$, we define a linear map $\hat{f}\in C^{k+l}(\g_1\oplus\g_2,\g_1\oplus\g_2)$ by
\begin{equation*}
 \hat{f}((x_1,v_1),(x_2,v_2),\cdots,(x_{k+l},v_{k+l}))=(\sum_{\tau\in\perm_{(k,l)}}\sgn(\tau)f(x_{\tau(1)},\cdots,x_{\tau(k)},v_{\tau(k+1)},\cdots,v_{\tau(k+l)}),0).
\end{equation*}
Similarly, for $f:\wedge^{k}\g_1\otimes\wedge^{l}\g_2\lon\g_2$, we define a linear map $\hat{f}\in C^{k+l}(\g_1\oplus\g_2,\g_1\oplus\g_2)$ by
\begin{equation*}
 \hat{f}((x_1,v_1),(x_2,v_2),\cdots,(x_{k+l},v_{k+l}))=(0,\sum_{\tau\in\perm_{(k,l)}}\sgn(\tau)f(x_{\tau(1)},\cdots,x_{\tau(k)},v_{\tau(k+1)},\cdots,v_{\tau(k+l)})).
\end{equation*}
The linear map $\hat{f}$ is called a {\bf lift} of $f$. For
example, the lifts of linear maps $\alpha:\wedge^2\g_1\lon\g_1$
and $\beta:\g_1\otimes\g_2\lon\g_2$ are  respectively
given by
\begin{eqnarray}
\label{semi-direct-1}\hat{\alpha}\big((x_1,v_1),(x_2,v_2)\big)&=&(\alpha(x_1,x_2),0),\\
\label{semi-direct-2}\hat{\beta}\big((x_1,v_1),(x_2,v_2)\big)&=&(0,\beta(x_1,v_2)-\beta(x_2,v_1)).
\end{eqnarray}

We define $\g^{k,l}=\wedge^{k}\g_1\otimes\wedge^{l}\g_2$. The vector space $\wedge^n(\g_1\oplus\g_2)$ is isomorphic to the direct sum of $\g^{k,l}$ with $k+l=n$.
\begin{defi} Let $\g_1$ and $\g_2$ be two vector
spaces. A linear map
$f\in\Hom(\wedge^{k+l+1}(\g_1\oplus\g_2),\g_1\oplus\g_2)$ has a
{\bf bidegree} $k|l$, if the following conditions hold:
\begin{itemize}
\item[\rm(i)] If $X$ is an element in $\g^{k+1,l}$, then $f(X)\in\g_1;$
\item[\rm(ii)] If $X$ is an element in $\g^{k,l+1}$, then $f(X)\in\g_2;$
\item[\rm(iii)] All the other case, $f(X)=0.$
\end{itemize}
\end{defi}
We call a linear map $f$ {\bf homogeneous} if $f$ has a bidegree and denote its bidegree by $||f||$. The linear maps $\hat{\alpha},~\hat{\beta}\in C^2(\g_1\oplus\g_2,\g_1\oplus\g_2)$ given by \eqref{semi-direct-1} and \eqref{semi-direct-2} have the bidegree  $||\hat{\alpha}||=||\hat{\beta}||=1|0$.
\emptycomment{Naturally we obtain a homogeneous linear map of the bidegree $1|0$,
\begin{eqnarray}
\label{semi-direct}\hat{\mu}:=\hat{\alpha}+\hat{\beta}.
\end{eqnarray}
Observe that $\hat{\mu}$ is a multiplication of the semi-direct product type,
$$
\hat{\mu}\big((x_1,v_1),(x_2,v_2)\big)=(\alpha(x_1,x_2),\beta(x_1,v_2)-\beta(x_2,v_1)).
$$
Even though $\hat{\mu}$ is not a lift (there is no $\mu$), we still use the symbol for our convenience below. \yh{delete something about the bidegree that is not used in the sequel.}}

The following lemma shows that the Nijenhuis-Richardson bracket on $C^*(\g_1\oplus\g_2,\g_1\oplus\g_2)$ is compatible with the
bigrading.
\begin{lem}\label{lem:bidegree perserve}
If $||f||=l_f|k_f$ and $||g||=l_g|k_g$, then $[f,g]_{\NR}$ has the bidegree $l_f+l_g|k_f+k_g.$
\end{lem}

The following result is well known. See the survey \cite{GLST} for more details.
\begin{pro}\label{pro:MC}
 Let $(\g,\pi)$ be a Lie algebra and $(V;\rho)$  a representation of $\g$. Then we have
 \begin{equation}
   [\hat{\pi}+\hat{\rho},\hat{\pi}+\hat{\rho}]_{\NR}=0.
 \end{equation}
\end{pro}

Let $(\g,\pi)$ be a Lie algebra and $(V;\rho)$  a representation
of $\g$. Note that $\hat{\pi}+\hat{\rho}\in C^{1\mid 0}(\g\oplus
V,\g\oplus V)$. Denote by $C^n(\g,V):=C^{n\mid -1}(\g\oplus
V,\g\oplus V)$ and $C^*(\g,V)=\oplus_{n=0}^{+\infty}C^n(\g,V)$.
 By the lift
map defined above, we have
$$C^n(\g,V):=C^{n\mid -1}(\g\oplus V,\g\oplus V)\cong\Hom(\wedge^n\g,V).$$
Define the coboundary operator $\dM^n_{\pi+\rho}:C^n(\g,V)\rightarrow C^{n+1}(\g,V)$ by
\begin{equation}\label{eq:CE-operator}
  \dM^n_{\pi+\rho} f:=(-1)^{n-1}[\hat{\pi}+\hat{\rho},\hat{f}]_{\NR},\quad \forall~f\in C^n(\g,V).
\end{equation}
In fact, since $\hat{\pi}+\hat{\rho}\in C^{1|0}(\g\oplus V,\g\oplus V)$ and $f\in C^{n|-1}(\g\oplus V,\g\oplus V)$, by Lemma \ref{lem:bidegree perserve}, we have $\dM^n_{\pi+\rho} f\in C^{n+1}(\g,V) $. By Proposition \ref{pro:MC}, we have $[\hat{\pi}+\hat{\rho},\hat{\pi}+\hat{\rho}]_{\NR}=0$. Because of the graded Jacobi identity, we have $\dM^{n+1}_{\pi+\rho}\circ \dM^n_{\pi+\rho}=0$. Thus we obtain a well-defined cochain complex $(C^*(\g,V),\dM^*_{\pi+\rho})$.

By a direct calculation, we have
\begin{pro}\label{pro:CE-operator} Let $(\g,\pi)$ be a Lie algebra and $(V;\rho)$  a representation of
$\g$. Then for all $f\in C^n(\g,V)$ and
$x_1,x_2,\cdots,x_{n+1}\in \g$, we have
  \begin{eqnarray*}
\dM^n_{\pi+\rho} f(x_1,\ldots,x_{n+1})&=&\sum_{i=1}^{n+1}(-1)^{i+1}\rho({x_i}) f(x_1,\ldots,\hat{x_i},\ldots,x_{n+1})\\
&&+\sum_{1\leq i<j\leq n+1}(-1)^{i+j}f(\pi(x_i,x_j),\ldots,\hat{x_{i}},\ldots,\hat{x_{j}},\ldots,x_{n+1}).
\end{eqnarray*}
 Thus the coboundary operator $\dM^n_{\pi+\rho}$ defined by \eqref{eq:CE-operator} is exactly the Chevalley-Eilenberg  coboundary operator for the Lie algebra $(\g, \pi)$ with coefficients in $(V;\rho)$.
\end{pro}

\section{Bidifferential graded Lie algebras, Maurer-Cartan characterizations and  cohomologies  of compatible Lie algebras} \label{sec:cohomology I}

In this section, we introduce the notions of  a bidifferential graded Lie algebra  and its
Maurer-Cartan elements. We construct the bidifferential graded Lie
algebra whose Maurer-Cartan elements are compatible Lie algebra
structures. We also   give a  cohomology
theory of a compatible Lie algebra, and use the second cohomology
group   to classify infinitesimal deformations.

\begin{defi}{\rm(\cite{GS1,GS2,OS1})} A {\bf compatible Lie algebra} is a triple $(\mathfrak{g},[-,-],\{-,-\})$, where $\g$ is a vector space,
   $ [-,-] $ and
$ \{-,-\} $ are Lie algebra structures on $\g$
such that
\begin{equation}[\{x,y\},z]+[\{y,z\},x]+[\{z,x\},y]+\{[x,y],z\}+\{[y,z],x\}+\{[z,x],y\}=0,\quad \forall x,y,z\in \mathfrak{g}.\label{eq:cl}\end{equation}
\end{defi}

\begin{pro}{\rm(\cite{Ku,S})}\label{eq:CL}
A triple $(\mathfrak{g},[-,-],\{-,-\})$ is a  compatible Lie
algebra if and only if
   $ [-,-] $ and
$ \{-,-\} $ are Lie algebra structures on $\g$ such that for any
$k_1,k_2\in\mathbb{K}$, the following bilinear operation
\begin{equation}\label{eq:Courant bracket}
\Courant{x,y}=k_1[x,y]+k_2\{x,y\},\quad \forall x,y\in\mathfrak{g}
\end{equation}
defines a Lie algebra structure on $\mathfrak{g}$.\end{pro}

\begin{defi}
  A {\bf homomorphism} between two compatible Lie algebras $(\mathfrak{g},[-,-]_\g,\{-,-\}_\g)$ and $(\mathfrak{h},[-,-]_\h,\{-,-\}_\h)$ is a linear map $\varphi:\g\rightarrow \mathfrak{\h}$ such that $\varphi$ is both a Lie algebra homomorphism between $(\mathfrak{g},[-,-]_\g)$ and $(\mathfrak{h},[-,-]_\h)$ and a Lie algebra homomorphism between $(\mathfrak{g},\{-,-\}_\g)$ and $(\mathfrak{h},\{-,-\}_\h)$.
\end{defi}

\begin{defi} {\rm (\cite{WB})}
 A {\bf representation} of  a compatible Lie algebra
$(\mathfrak{g},[-,-],\{-,-\})$ on a
vector space $V$  consists of a pair of linear maps
 $\rho, \mu: \mathfrak{g}\rightarrow \mathfrak{gl}(V)$ such that $\rho$ is a representation of the Lie algebra $(\mathfrak{g},[-,-])$ on $V$, $\mu$ is a representation of the Lie algebra $(\mathfrak{g}, \{-,-\})$ on $V$, and
 \begin{eqnarray}
  \rho(\{x,y\})+\mu([x,y])&=&[\rho(x),\mu(y)]-[\rho(y),\mu(x)].\label{eq:r3}
 \end{eqnarray}
\end{defi}

\begin{ex}{\rm Let $(\mathfrak{g},[-,-],\{-,-\})$ be a compatible Lie algebra. Define $\ad,\AD:\g\lon\gl(\g)$ by $\ad_xy=[x,y]$ and $\AD_xy=\{x,y\}$ for all $x,y\in \g$. Then $(\g;\ad,\AD)$ is a representation of $(\mathfrak{g},[-,-],\{-,-\})$, which is called the {\bf adjoint representation}.
   }
 \end{ex}

\subsection{Bidifferential graded Lie algebras and  Maurer-Cartan characterizations of compatible Lie algebras}

In this subsection, we introduce the notions of a bidifferential
graded  Lie algebra and its Maurer-Cartan elements. In particular,
we give the Maurer-Cartan characterizations of compatible Lie
algebras and the bidifferential graded  Lie algebra that controls
deformations of  a compatible Lie
algebra.
\begin{defi}
  Let $(L,[-,-],\partial_1)$ and $(L,[-,-],\partial_2)$ be two differential graded Lie algebras. We call the quadruple $(L,[-,-],\partial_1,\partial_2)$ a {\bf bidifferential graded Lie algebra} if $\partial_1$ and $\partial_2$
  satisfy
  \begin{equation}
  \partial_1\circ \partial_2+\partial_2\circ \partial_1=0.
  \end{equation}
\end{defi}

The following proposition gives an equivalent description of
 a  bidifferential graded Lie algebra.
\begin{pro}
Let $(L,[-,-],\partial_1)$ and $(L,[-,-],\partial_2)$ be two differential graded Lie algebras. Then  $(L,[-,-],\partial_1,\partial_2)$ is a  bidifferential graded Lie algebra if and only if for any $k_1,k_2\in\K$, $(L,[-,-],\partial_{k_1,k_2})$ is a differential graded Lie algebra, where $\partial_{k_1,k_2}=k_1\partial_1+k_2\partial_2$.
\end{pro}
\begin{proof}
Assume that $(L,[-,-],\partial_1,\partial_2)$ is a  bidifferential graded Lie algebra. Then we have
\begin{eqnarray*}
    \partial_{k_1,k_2}\circ \partial_{k_1,k_2}&=&(k_1\partial_1+k_2\partial_2)\circ (k_1\partial_1+k_2\partial_2)\\
    &=&k_1^2\partial_1\circ \partial_1+k_1k_2(\partial_1\circ \partial_2+\partial_2\circ \partial_1)+k^2_2\partial_2\circ \partial_2=0.
    \end{eqnarray*}
     Furthermore, for any homogenous elements $x,y\in L$, we have
    \begin{eqnarray*}
        \partial_{k_1,k_2}[x,y]&=&(k_1\partial_1+k_2\partial_2)[x,y]\\
        &=&k_1([\partial_1 x,y]+(-1)^{|x|}[x,\partial_1y])+k_2([\partial_2 x,y]+(-1)^{|x|}[x,\partial_2y])\\
        &=&[(k_1\partial_1+k_2\partial_2) x,y]+(-1)^{|x|}[x,(k_1\partial_1+k_2\partial_2)y] \\
        &=&[\partial_{k_1,k_2} x,y]+(-1)^{|x|}[x,\partial_{k_1,k_2}y].\end{eqnarray*}
Thus for any $k_1,k_2\in\K$, $(L,[-,-],\partial_{k_1,k_2})$ is a differential graded Lie algebra.

The converse can be proved similarly. We omit the details.
\end{proof}

\begin{defi}
Let $(L,[-,-],\partial_1,\partial_2)$ be a bidifferential graded Lie algebra. A pair $(P_1,P_2)\in L_1\oplus L_1$ is called a {\bf Maurer-Cartan element} of the bidifferential graded Lie algebra $(L,[-,-],\partial_1,\partial_2)$ if $P_1$ and $P_2$ are Maurer-Cartan elements of the  differential graded Lie algebras $(L,[-,-],\partial_1)$ and  $(L,[-,-],\partial_2)$ respectively, and
\begin{eqnarray}
\partial_1 P_2+\partial_2 P_1 +[P_1,P_2]=0.
\end{eqnarray}
\end{defi}

\begin{pro}\label{pro:MC-equivalent 2-para}
  Let $(L,[-,-],\partial_1,\partial_2)$ be a bidifferential graded Lie algebra. A pair $(P_1,P_2)$ is a Maurer-Cartan element of  $(L,[-,-],\partial_1,\partial_2)$ if and only if for any $k_1,k_2\in\K$, $k_1 P_1+k_2 P_2$ is a Maurer-Cartan element of the differential graded Lie algebra $(L,[-,-],\partial_{k_1,k_2})$.
\end{pro}
\begin{proof}
Assume that the pair $(P_1,P_2)$ is a Maurer-Cartan element of  $(L,[-,-],\partial_1,\partial_2)$. Then for any $k_1,k_2\in\K$, we have
\begin{eqnarray*}
&&\partial_{k_1,k_2}(k_1P_1+k_2 P_2)+\half[k_1P_1+k_2 P_2,k_1P_1+k_2 P_2]\\
&=&k_1^2(\partial_1 P_1+\frac{1}{2}[P_1,P_1])+k_1k_2(\partial_1
P_2+\partial_2 P_1 +[P_1,P_2])+k_2^2(\partial_2
P_2+\frac{1}{2}[P_2,P_2])=0.\end{eqnarray*} Thus  $k_1 P_1+k_2
P_2$ is a Maurer-Cartan element of the differential graded Lie
algebra $(L,[-,-],\partial_{k_1,k_2})$.

The converse can be proved similarly. We omit the details.
\end{proof}

Let $(L,[-,-])$ be a   graded Lie algebra. It is obvious that $(L,[-,-],\partial_1=0,\partial_2=0)$ is a bidifferential graded Lie algebra. Consider the graded Lie algebra $(C^*(\g,\g),[-,-]_{\NR})$, we obtain the following main result.

\begin{thm}\label{pro:lsymNR}
  Let $\g$ be a vector space, $\pi_1,\pi_2\in\Hom(\wedge^2\g,\g)$. Then $(\g,\pi_1,\pi_2)$ is a compatible Lie algebra if and only if $(\pi_1,\pi_2)$ is a Maurer-Cartan element of the bidifferential graded Lie algebra $(C^*(\g,\g),[-,-]_{\NR}, \partial_1=0,\partial_2=0)$.
\end{thm}
\begin{proof}
By Proposition \ref{pro:MCLie},
$\pi_1,\pi_2\in\Hom(\wedge^2\g,\g)$ define Lie algebra structures
on $\g$   respectively  if and only if
$$
[\pi_1,\pi_1]_\NR=0,\quad[\pi_2,\pi_2]_\NR=0.
$$
Moreover, $[\pi_1,\pi_2]_{\NR}=0$ exactly corresponds to the compatibility condition (\ref{eq:cl}).
  Thus, $(\g,\pi_1,\pi_2)$ is a compatible Lie algebra if and only if
 \begin{equation}\label{eq:Maurer-Cartatn 1}
 [\pi_1,\pi_1]_{\NR}=0,\quad [\pi_1,\pi_2]_{\NR}=0,\quad[\pi_2,\pi_2]_{\NR}=0,
 \end{equation}
which means that  $(\pi_1,\pi_2)$ is a Maurer-Cartan element of the bidifferential graded Lie algebra $(C^*(\g,\g),[-,-]_{\NR}, \partial_1=0,\partial_2=0)$.
\end{proof}

By twisting with a Maurer-Cartan element, we can obtain a new bidifferential graded Lie algebra, which controls deformations the original Maurer-Cartan element.
\begin{pro}\label{thm:main proposition}
Let  $(P_1,P_2)$ be a Maurer-Cartan element of a bidifferential
graded Lie algebra $(L,[-,-],\partial_1,\partial_2)$. Then
$(L,[-,-],\tilde{\partial}_1,\tilde{\partial}_2)$ is a
bidifferential graded Lie algebra, where $\tilde{\partial}_1$ and
$\tilde{\partial}_2$ are defined by
$$\tilde{\partial}_1 Q=\partial_1Q+[P_1,Q],\quad\tilde{\partial}_2 Q=\partial_2 Q+[P_2,Q],\quad \forall~Q\in L.$$
 For any $\tilde{P}_1,\tilde{P}_2\in L_1$,
$(P_1+\tilde{P}_1,P_2+\tilde{P}_2)$ is a Maurer-Cartan element of
the bidifferential graded Lie algebra
$(L,[-,-],\partial_1,\partial_2)$ if and only if
$(\tilde{P}_1,\tilde{P}_2)$ is a Maurer-Cartan element of the
bidifferential graded Lie algebra
$(L,[-,-],\tilde{\partial}_1,\tilde{\partial}_2)$.
\end{pro}
\begin{proof}
Since $P_1$ is a Maurer-Cartan element of the differential graded Lie algebra $(L,[-,-],\partial_1)$,  $(L,[-,-],\tilde{\partial}_1)$ is a differential graded Lie algebra. Similarly,  $(L,[-,-],\tilde{\partial}_2)$ is a differential graded Lie algebra. Since  $\partial_1P_2+\partial_2P_1+[P_1,P_2]=0$, by properties of $\partial_1$ and $\partial_2$ and the graded Jacobi identity,  for any $Q\in L$, we have
\begin{eqnarray*}
    \tilde{\partial}_1\tilde{\partial}_2Q+\tilde{\partial}_2\tilde{\partial}_1Q
    &=&\partial_1\partial_2 Q+[P_1,\partial_2 Q]+\half \partial_1[P_2,Q]+\half [P_1,[P_2,Q]]\\
    &&+\partial_2\partial_1 Q+[P_2,\partial_1Q]+\half \partial_2[P_1,Q]+\half [P_2,[P_1,Q]]\\
    &=&\half [\partial_2 P_1,Q]+\half [\partial_1 P_2,Q]+\half[[P_1,P_2],Q]=0,
    \end{eqnarray*}
    which implies that $\tilde{\partial}_1\tilde{\partial}_2+\tilde{\partial}_2\tilde{\partial}_1=0$.
Thus $(L,[-,-],\tilde{\partial}_1,\tilde{\partial}_2)$ is a bidifferential graded Lie algebra.

If $(P_1+\tilde{P}_1,P_2+\tilde{P}_2)$ is a Maurer-Cartan element of the bidifferential graded Lie algebra $(L,[-,-],\partial_1,\partial_2)$, then we have
\begin{eqnarray*}
\partial_1 (P_1+\tilde{P}_1)+\half [P_1+\tilde{P}_1,P_1+\tilde{P}_1]&=&0,\\
 \partial_2 (P_2+\tilde{P}_2)+  \half[P_2+\tilde{P}_2,P_2+\tilde{P}_2]&=&0,\\
 \partial_1 (P_2+\tilde{P}_2)+\partial_2 (P_1+\tilde{P}_1)+[P_1+\tilde{P}_1,P_2+\tilde{P}_2]&=&0.
\end{eqnarray*}
Furthermore, since $(P_1,P_2)$ is a Maurer-Cartan element of the bidifferential graded Lie algebra $(L,[-,-],\partial_1,\partial_2)$, we have
\begin{eqnarray*}
\tilde{\partial}_1 \tilde{P}_1+\frac{1}{2}[\tilde{P}_1,\tilde{P}_1]&=&0,\\
\tilde{\partial}_2 \tilde{P}_2+\frac{1}{2}[\tilde{P}_2,\tilde{P}_2]&=&0,\\
\tilde{\partial}_1 \tilde{P}_2+\tilde{\partial}_2 \tilde{P}_1+[\tilde{P}_1,\tilde{P}_2]&=&0.
\end{eqnarray*}
Thus $(\tilde{P}_1,\tilde{P}_2)$ is a Maurer-Cartan element of the bidifferential graded Lie algebra $(L,[-,-],\tilde{\partial}_1,\tilde{\partial}_2)$. The converse can be proved similarly. We omit the details.
\end{proof}

Now we are ready to give the bidifferential graded Lie algebra that controls deformations of a compatible Lie algebra.

\begin{thm}\label{pro:new differential Lie algebra}
Let $(\g,\pi_1,\pi_2)$ be a compatible Lie algebra.  Then
we have the following conclusions.
\begin{itemize}
\item[\rm (1)]$(C^*(\g,\g),[-,-]_\NR,
\tilde{\partial}_1,\tilde{\partial}_2)$ is a bidifferential graded
Lie algebra, where $\tilde{\partial}_1$ and $\tilde{\partial}_2$
are  respectively  defined by
$$\tilde{\partial}_1 ={[\pi_1,Q]}_{\NR},\quad\tilde{\partial}_2 ={[\pi_2,Q]}_{\NR},\quad\forall~Q\in C^q(\g,\g).$$
\item[\rm(2)] For any
$\tilde{\pi}_1,\tilde{\pi}_2\in\Hom(\wedge^2\g,\g)$,
$(\g,\pi_1+\tilde{\pi}_1,\pi_2+\tilde{\pi}_2)$ is  a compatible
Lie algebra if and only if the pair
$(\tilde{\pi}_1,\tilde{\pi}_2)$ is a Maurer-Cartan element of the
bidifferential graded Lie algebra $(C^*(\g,\g),[-,-]_\NR,
\tilde{\partial}_1,\tilde{\partial}_2)$.
\end{itemize}
 \end{thm}
\begin{proof}
 By  Theorem \ref{pro:lsymNR},   $(\pi_1,\pi_2)$ is a Maurer-Cartan element of the bidifferential graded Lie algebra $(C^*(\g,\g),[-,-]_{\NR},\partial_1=0,\partial_2=0)$. By Proposition \ref{thm:main proposition}, we obtain the conclusions.
\end{proof}

\subsection{Cohomologies  of compatible Lie algebras}
In this subsection, we introduce a cohomology theory of compatible Lie algebras.

Let $(\mathfrak{g},[-,-],\{-,-\})$ be a compatible Lie algebra with $\pi_1(x,y)=[x,y]$ and $\pi_2(x,y)=\{x,y\}$. By Theorem \ref{pro:lsymNR},  $(\pi_1,\pi_2)$ is a Maurer-Cartan element of the bidifferential graded Lie algebra $(C^*(\g,\g),[-,-]_{\NR},\partial_1=0,\partial_2=0)$.
Recall that $C^n(\g,\g)=\Hom(\wedge^n\g,\g)$.

Define the space of 0-cochains $\huaC^0(\g,\g)$ by $\huaC^0(\g,\g)=\{x\in\g\mid [x,y]=\{x,y\},~\forall~y\in\g\}$. For $n\ge 1$, define the space of $n$-cochains $\huaC^n(\g,\g)$ by
$$\huaC^n(\g,\g)=\underbrace{C^n(\g,\g)\oplus C^n(\g,\g)\cdots\oplus C^n(\g,\g)}_{n~{copies}}.$$

Note that $[\pi_1,x]_\NR(y)=-[x,y]$ and $[\pi_2,x]_\NR(y)=-\{x,y\}$. Define $\jetd^0:\huaC^{0}(\g,\g)\longrightarrow \huaC^{1}(\g,\g)$  by
$$\jetd^0 x=-[\pi_1,x]_\NR=-[\pi_2, x]_\NR,\quad \forall x\in \huaC^0(\g,\g),$$
define $\jetd^1:\huaC^{1}(\g,\g)\longrightarrow \huaC^{2}(\g,\g)$  by
$$\jetd^1 f=([\pi_1, f]_\NR,[\pi_2, f]_\NR),\quad \forall f\in \huaC^1(\g,\g)$$
and define $\jetd^n:\huaC^{n}(\g,\g)\longrightarrow \huaC^{n+1}(\g,\g)$  by
\begin{eqnarray*}
    \jetd^n(\omega_1,\cdots,\omega_{n})=(-1)^{n-1}([\pi_1,\omega_1]_\NR,\cdots,\underbrace{[\pi_2,\omega_{i-1}]_\NR+[\pi_1,\omega_i]_\NR}_i,\cdots,[\pi_2,\omega_{n}]_\NR),
\end{eqnarray*}
where $(\omega_1,\omega_2,\cdots,\omega_n)\in\huaC^n(\g,\g)$ and $2\leq i\leq n$.

The definition of $\jetd$ can be presented by the following diagram:

{\footnotesize{\begin{tikzpicture}[>=stealth,sloped]
    \matrix (tree) [%
      matrix of nodes,
      minimum size=0.2cm,
      column sep=1.6cm,
      row sep=0.2cm,
    ]
    {     & & & &$C^4(\g,\g)$ &\\
          & &  & $C^3(\g,\g)$ & &\\
          & &$C^2(\g,\g)$ &   & $C^4(\g,\g)$ &\\
      $\huaC^0(\g,\g)$ & $C^1(\g,\g)$  &  & $C^3(\g,\g)$ & &$\cdots$\\
         & & $C^2(\g,\g)$ &   & $C^4(\g,\g)$ &\\
         & &   & $C^3(\g,\g)$ & &\\
         & & & & $C^4(\g,\g)$ &\\
    };
    \draw[->] (tree-4-1) -- (tree-4-2) node [midway,above] {$-[\pi_1,-]_\NR$};
    \draw[->] (tree-4-2) -- (tree-3-3) node [midway,above] {$[\pi_1,-]_\NR$};
    \draw[->] (tree-4-2) -- (tree-5-3) node [midway,below] {$[\pi_2,-]_\NR$};
    \draw[->] (tree-3-3) -- (tree-2-4) node [midway,above] {$-[\pi_1,-]_\NR$};
    \draw[->] (tree-3-3) -- (tree-4-4) node [midway,above] {$-[\pi_2,-]_\NR$};
    \draw[->] (tree-5-3) -- (tree-4-4) node [midway,above] {$-[\pi_1,-]_\NR$};
    \draw[->] (tree-5-3) -- (tree-6-4) node [midway,below] {$-[\pi_2,-]_\NR$};
    \draw[->] (tree-2-4) -- (tree-1-5) node [midway,above] {$[\pi_1,-]_\NR$};
    \draw[->] (tree-2-4) -- (tree-3-5) node [midway,above] {$[\pi_2,-]_\NR$};
    \draw[->] (tree-4-4) -- (tree-3-5) node [midway,above] {$[\pi_1,-]_\NR$};
    \draw[->] (tree-4-4) -- (tree-5-5) node [midway,above] {$[\pi_2,-]_\NR$};
    \draw[->] (tree-6-4) -- (tree-5-5) node [midway,above] {$[\pi_1,-]_\NR$};
    \draw[->] (tree-6-4) -- (tree-7-5) node [midway,below] {$[\pi_2,-]_\NR$};
  \end{tikzpicture}}}

\begin{thm}\label{thm:cohomology of CLA}
  With the above notations, we have $\jetd^{n+1}\circ\jetd^n=0$, i.e. $(\huaC^*(\g,\g)=\oplus_{n=0}^\infty \huaC^n(\g,\g),\jetd^*)$ is a cochain complex.
\end{thm}
\begin{proof}
By the fact that $[\pi_1,\pi_1]_\NR=[\pi_2,\pi_2]_\NR=0$ and the graded Jacobi identity,  we have $\jetd^1\circ\jetd^0=0$.   Then by \eqref{eq:Maurer-Cartatn 1} and the graded Jacobi identity, for $(\omega_1,\omega_2,\cdots,\omega_n)\in\huaC^n(\g,\g),~n\ge 1$, we have
\begin{eqnarray*}
    &&\jetd^{n+1}\jetd^n(\omega_1,\cdots,\omega_{n})\\
    &=&(-1)^{n-1}\delta^{n+1}([\pi_1,\omega_1]_\NR,\cdots,\underbrace{[\pi_2,\omega_{i-1}]_\NR+[\pi_1,\omega_i]_\NR}_i,\cdots,[\pi_2,\omega_{n}]_\NR)\\
    &=&-([\pi_1,[\pi_1,\omega_1]_\NR]_\NR,[\pi_2,[\pi_1,\omega_1]_\NR]_\NR+[\pi_1,[\pi_2,\omega_1]_\NR]_\NR+[\pi_1,[\pi_1,\omega_2]_\NR]_\NR,\cdots,\\
    &&\underbrace{[\pi_2,[\pi_2,\omega_{i-2}]_\NR]_\NR+[\pi_2,[\pi_1,\omega_{i-1}]_\NR]_\NR+[\pi_1,[\pi_2,\omega_{i-1}]_\NR]_\NR+[\pi_1,[\pi_1,\omega_{i}]_\NR]_\NR}_{3\leq i\leq n-1},\cdots,\\
&&[\pi_2,[\pi_2,\omega_{n-1}]_\NR]_\NR+[\pi_2,[\pi_1,\omega_n]_\NR]_\NR+[\pi_1,[\pi_2,\omega_n]_\NR]_\NR,[\pi_2,[\pi_2,\omega_n]_\NR]_\NR)\\
    &=&-(\half[[\pi_1,\pi_1]_\NR,\omega_1]_\NR,[[\pi_1,\pi_2]_\NR,\omega_1]_\NR+\half [[\pi_1,\pi_1]_\NR,\omega_2]_\NR,\cdots,\\
    &&\underbrace{\half[[\pi_2,\pi_2]_\NR,\omega_{i-2}]_\NR+[[\pi_1,\pi_2]_\NR,\omega_{i-1}]_\NR+\half[[\pi_1,\pi_1]_\NR,\omega_i]_\NR}_{3\leq i\leq n-1},\cdots,\\
    &&\half[[\pi_2,\pi_2]_\NR,\omega_{n-1}]_\NR+[[\pi_1,\pi_2]_\NR,\omega_n]_\NR,\half[[\pi_2,\pi_2]_\NR,\omega_n]_\NR)\\
    &=&(0,0,\cdots,0).
    \end{eqnarray*}
    Thus we have $\jetd^{n+1}\circ\jetd^n=0$.\end{proof}
\begin{defi} Let $(\mathfrak{g},[-,-],\{-,-\})$ be a compatible Lie
algebra.
  The cohomology of the cochain complex $(\huaC^*(\g,\g),\jetd^*)$  is  called {\bf the cohomology of   $(\g,[-,-],\{-,-\})$}.  We denote the $n$-th cohomology group by  $\huaH^n(\g,\g)$.
\end{defi}

  Let $(\mathfrak{g},[-,-],\{-,-\})$ be a compatible
Lie algebra.  Clearly, $$\huaH^0(\g,\g)=\{x\in \g\mid
[x,y]=\{x,y\}=0,~\forall~y\in\g\}.$$ Thus $\huaH^0(\g,\g)$ is both
the center of the Lie algebra $(\g,[-,-])$ and  the Lie algebra
$(\g,\{-,-\})$.

A linear map $f\in\Hom(\g,\g)$ is called a {\bf derivation} of   $(\mathfrak{g},[-,-],\{-,-\})$ if
$$f([x,y])=[f(x),y]+[x,f(y)],  \quad f(\{x,y\})=\{f(x),y\}+\{x,f(y)\},~ \quad \forall~x,y\in \g.$$
A derivation $f$ is called an inner derivation   of
$(\mathfrak{g},[-,-],\{-,-\})$  if there exists $x\in {\g}$ such
that
$$f(y)=[x,y]=\{x,y\}, ~\forall~y\in \g.$$
 We denote by $\Der(\g,\g)$ and $ {\IDer}(\g,\g)$ the space of derivations and inner derivations respectively.
Then the first cohomology group $\huaH^1(\g,\g)=\Der(\g,\g)/\IDer(\g,\g).$

\subsection{Infinitesimal deformations of compatible Lie algebras}

In this subsection, we study infinitesimal deformations of compatible Lie algebras using the cohomology theory of compatible Lie algebras. In particular, we introduce the notion of a Nijenhuis operator on a compatible Lie algebra which gives rise to a trivial infinitesimal deformation.
\begin{defi}
  Let $(\mathfrak{g},[-,-],\{-,-\})$ be a compatible Lie algebra,
$(\omega_1,\omega_2)\in \huaC^{2}(\g,\g)$.
Define
\begin{equation}
[x,y]_t=[x,y]+t\omega_1(x,y),\;\;\{x,y\}_t=\{x,y\}+t\omega_2(x,y),\;\;\forall
x,y\in \mathfrak{g}.
\end{equation}
If for any $t$, $(\mathfrak{g},[-,-]_t,\{-,-\}_t)$ is still a
compatible Lie algebra, then we say that $(\omega_1,\omega_2)$
  generates  an {\bf infinitesimal deformation} of
$(\mathfrak{g},[-,-],\{-,-\})$.
\end{defi}
It is straightforward to verify that $(\omega_1,\omega_2)$
generates an infinitesimal deformation of a
compatible Lie algebra
 $(\frkg,[-,-],\{-,-\})$ if and only if for any $k_1,k_2\in
 \mathbb K$, $k_1\omega_1+k_2\omega_2$ generates an infinitesimal
 deformation of the Lie algebra  $(\frak g,
 \Courant{-,-}=k_1[-,-]+k_2\{-,-\})$.

We set
$$\pi_1(x,y)=[x,y],\quad \pi_2(x,y)=\{x,y\}.$$
By Theorem \ref{pro:lsymNR},  $(\g,[-,-]_t,\{-,-\}_t)$ is an infinitesimal deformation of $(\g,\pi_1,\pi_2)$  if and only if
\begin{equation}\label{eq:2-closed omega bracket}
 \begin{array}{rclrclrcl}
[{\pi}_1,{\omega}_1]_{\NR}&=&0,&[{\pi}_1,{\omega}_2]_{\NR}+[{\pi}_2,{\omega}_1]_{\NR}&=&0,&[{\pi}_2,{\omega}_2]_{\NR}&=&0,\\
{[{\omega}_1,{\omega}_1]}_{\NR}&=&0,& {[{\omega}_1,{\omega}_2]}_{\NR}&=&0,&{[{\omega}_2,{\omega}_2]}_{\NR}&=&0.
 \end{array}
\end{equation}
The first line of \eqref{eq:2-closed omega bracket} means that
$(\omega_1,\omega_2)\in \huaC^2(\g,\g)$  is a $2$-cocycle for the
compatible Lie algebra $(\mathfrak{g},[-,-],\{-,-\})$, i.e.
$\jetd^2(\omega_1,\omega_2)=0$, and the second line of
\eqref{eq:2-closed omega bracket} means that
$(\g,\omega_1,\omega_2)$ is a compatible Lie algebra. Moreover,
 by these equations  it is obvious that both $(\g,
\pi_1,\omega_1)$ and $(\g, \pi_2,\omega_2)$ are compatible Lie
algebras.

\begin{defi}
Two infinitesimal deformations $(\g,[-,-]_t,\{-,-\}_t)$ and
$(\g,[-,-]_t',\{-,-\}_t')$ of a compatible Lie algebra
$(\g,[-,-],\{-,-\})$ generated by $(\omega_1,\omega_2)$ and
$(\omega'_1,\omega'_2) $   respectively  are said to be
{\bf equivalent} if there exists $N\in \gl(\g)$ such that
${\Id}+tN: (\g,[-,-]_t',\{-,-\}_t')\longrightarrow
(\g,[-,-]_t,\{-,-\}_t)$ is a compatible Lie algebra homomorphism.
\end{defi}

Two infinitesimal deformations $(\g,[-,-]_t,\{-,-\}_t)$ and
$(\g,[-,-]_t',\{-,-\}_t')$ generated by $(\omega_1,\omega_2)$ and
$(\omega'_1,\omega'_2) $    respectively  are equivalent
if and only if
\begin{eqnarray}
\omega_1(x,y)-\omega_1'(x,y)&=&[x, N(y)]+[N(x), y]-N([x, y]),\label{2-exact1}\\
N\omega_1(x,y)&=&\omega_1'(x,N(y))+\omega_1'(N(x),y)+[N(x),N(y)],\label{integral condition 1}\\
\omega_2(x,y)-\omega_2'(x,y)&=&\{x, N(y)\}+\{N(x), y\}-N(\{x, y\}),\label{2-exact2}\\
N\omega_2(x,y)&=&\omega_2'(x,N(y))+\omega_2'(N(x),y)+\{N(x),N(y)\},\label{integral condition 2}\\
\omega_1'(N(x),N(y))&=&0,\label{eq:con1}\\
\omega_2'(N(x),N(y))&=&0.\label{eq:con2}
\end{eqnarray}
It is easy to see that $(\ref{2-exact1})$ and $(\ref{2-exact2})$ mean that $(\omega_1- \omega_1',\omega_2-\omega_2')=\jetd^1 N$.

We summarize the above discussion into the following
 conclusion:
\begin{thm}\label{thm:deformation}
Let $(\g,[-,-]_t,\{-,-\}_t)$ be an infinitesimal deformation  of a compatible Lie algebra $(\g,[-,-],\{-,-\})$ generated by $(\omega_1,\omega_2)$. Then $(\omega_1,\omega_2)$ is
closed, i.e. $\jetd^2(\omega_1,\omega_2)=0.$

Furthermore, if two infinitesimal deformations
$(\g,[-,-]_t,\{-,-\}_t)$ and $(\g,[-,-]_t',\{-,-\}_t')$ of a
compatible Lie algebra $(\g,[-,-],\{-,-\})$ generated by
$(\omega_1,\omega_2)$ and $(\omega'_1,\omega'_2) $
respectively  are equivalent,  then $(\omega_1,\omega_2)$ and
$(\omega_1',\omega'_2)$ are in the same cohomology class in
$\huaH^2(\g,\g)$.
\end{thm}

In the sequel, we recall Nijenhuis operators on Lie algebras and introduce the notion of Nijenhuis operators on compatible Lie algebras.
Let $(\g,[-,-])$ be a Lie algebra and $N:\g\longrightarrow \g$  a linear map. We use $\pi:\wedge^2\g\longrightarrow \g$ to denote the Lie algebra structure, i.e. $\pi(x,y)=[x, y]$. The deformed structure
$$\pi_N:=[\pi,N]_{\NR}$$
defines a bilinear operation on $\g$ which we denote by $[-,-]_N$. Explicitly,
\begin{equation}
[x, y]_N=[\pi,N]_{\NR}(x,y)=[N(x), y]+[x, N(y)]-N([x, y]),\quad \forall~x,y\in \g.
\end{equation}

We use $T_\pi N:\g\otimes\g\longrightarrow\g$ to denote the
Nijenhuis torsion of $N$ defined by
\begin{equation}\label{eq:LANij1}
T_\pi N:=\frac{1}{2}([\pi,N\circ
N]_{\NR}+[N,[\pi,N]_{\NR}]_{\NR}),\quad\forall x,y\in
\g.
\end{equation}
It is well-known that
  \begin{equation}\label{eq:TN}
T_\pi N(x,y)=N([x, y]_N)-[N(x),N(y)].
\end{equation}

\begin{defi}  Let $(\g,[-,-])$ be a Lie algebra. A linear map $N:\g\longrightarrow
\g$ is called a {\bf Nijenhuis operator}   on
$(\g,[-,-])$ if $T_\pi N=0$, i.e.
\begin{equation}\label{eq:Lie-Nijenhuis}
N([N(x),y]+[x,N(y)]-N([x,y]))=[N(x),N(y)],\quad
\forall~x,y\in\g.
\end{equation}
\end{defi}

Nijenhuis operators on Lie algebras have the following important
property.
 \begin{pro}{\rm(\cite{Dorf})}\label{eq:Nijenhuis property}
 If $N:\g\longrightarrow \g$ is a Nijenhuis operator on a Lie algebra $(\g,[-,-])$, then $(\g,[-, -]_N)$ is also a Lie algebra and $N$ is a  Lie algebra
 homomorphism from $(\g,[-, -]_N)$ to $(\g,[-,-])$.
 \end{pro}

\begin{defi}
Let $(\mathfrak{g},[-,-],\{-,-\})$ be a
compatible Lie algebra. A linear map $N: \g\rightarrow
\mathfrak{g}$  is called a {\bf Nijenhuis operator} on
$(\mathfrak{g},[-,-],\{-,-\})$ if $N$ is both a Nijenhuis operator on the Lie algebra
$(\mathfrak{g},[-,-])$ and a Nijenhuis operator on the Lie algebra
$(\mathfrak{g},\{-,-\})$.
\end{defi}

\begin{pro}\label{pro:Nijenhuis torsion}
 Let $(\g,[-,-],\{-,-\})$ be a compatible Lie algebra with $\pi_1(x,y)=[x,y]$ and $\pi_2(x,y)=\{x,y\}$. Let $N:\g\rightarrow \g$ be a linear map.
 Then $N$ is a Nijenhuis operator on the compatible Lie algebra $(\g,[-,-],\{-,-\})$ if and only if for any $k_1,k_2\in\K$, the Nijenhuis torsion of $N$ for the Lie algebra
 structure $\pi=k_1\pi_1+k_2\pi_2$ vanishes, i.e.
  for any $k_1,k_2\in\K$, $N$ is a Nijenhuis operator on the Lie algebra  $(\frak g,
 \Courant{-,-}=k_1[-,-]+k_2\{-,-\})$.
\end{pro}
\begin{proof}
  By a direct calculation, we have
\begin{equation*}\label{eq:Nijenhui torsion}
T_{\pi}N=k_1T_{\pi_1}N+k_2T_{\pi_2}N.
\end{equation*}
Then $T_{\pi}N=0$ if and only if $T_{\pi_1}N=T_{\pi_2}N=0$. Thus $T_{\pi}N=0$ if and only if $N$ is a Nijenhuis operator on the compatible Lie algebra $(\g,[-,-],\{-,-\})$.
\end{proof}

\begin{pro}\label{pro:Nijenhuis operator property}
 Let $(\mathfrak{g},[-,-],\{-,-\})$ be a compatible Lie algebra and $N:\g\rightarrow \g$   a Nijenhuis operator. Then $(\g,[-,-]_N,\{-,-\}_N)$ is also a compatible Lie algebra and $N$ is a compatible Lie algebra
 homomorphism from
 $(\g,[-,-]_N,\{-,-\}_N)$  to $(\g,[-,-],\{-,-\})$.
\end{pro}
\begin{proof}
  Let  $N:\g\rightarrow \g$ be a Nijenhuis operator on the compatible Lie algebra $(\mathfrak{g},[-,-],\{-,-\})$. By Proposition \ref{pro:Nijenhuis torsion}, $N$ is a Nijenhuis operator on the Lie algebra $(\g,\Courant{-,-})$, where the bracket $\Courant{-,-}$ is given by \eqref{eq:Courant bracket}. By a direct calculation, we have
  \begin{eqnarray*}
    \Courant{x,y}_N&=&\Courant{N(x),y}+\Courant{x,N(y)}-N\Courant{x,y}\\
    &=&k_1[x,y]_N+k_2\{x,y\}_N.
  \end{eqnarray*}
By Proposition \ref{eq:Nijenhuis property},
$(\g,[-,-]_N,\{-,-\}_N)$ is a compatible Lie algebra and $N$ is a
 compatible
Lie algebra
 homomorphism from $(\g,[-,-]_N,\{-,-\}_N)$  to $(\g,[-,-],\{-,-\})$.
\end{proof}

\begin{defi}
An infinitesimal deformation $(\g,[-,-]_t,\{-,-\}_t)$ of a compatible Lie algebra $(\g,[-,-],\{-,-\})$ generated by $(\omega_1,\omega_2)$ is {\bf trivial} if   there exists $N\in\gl(\g)$ such that ${\Id}+tN:(\g,[-,-]_t,\{-,-\}_t)\longrightarrow (\g,[-,-],\{-,-\})$ is a compatible Lie algebra homomorphism.
\end{defi}

One can deduce that $(\mathfrak{g},[-,-]_t,\{-,-\}_t)$ is a  trivial infinitesimal deformation if and only if
\begin{eqnarray}
\omega_1(x,y)&=&[x, N(y)]+[N(x), y]-N([x, y]),\label{Nij1}\\
N(\omega_1(x,y))&=&[N(x), N(y)],\label{Nij2}\\
\omega_2(x,y)&=&\{x, N(y)\}+\{N(x), y\}-N(\{x, y\}),\label{Nij3}\\
N(\omega_2(x,y))&=&\{N(x), N(y)\}.\label{Nij4}
\end{eqnarray}
By $(\ref{Nij1})$ and $(\ref{Nij2})$, $N$ is a Nijenhuis operator on the Lie algebra $(\g,[-,-])$.  By $(\ref{Nij3})$ and $(\ref{Nij4})$, $N$ is a Nijenhuis operator on the Lie algebra $(\g,\{-,-\})$. Thus $N$ is a Nijenhuis operator on the compatible Lie algebra $(\g,[-,-],\{-,-\})$.

We have seen that a trivial infinitesimal deformation of a compatible Lie algebra gives rise to a Nijenhuis operator. In fact, the converse is also true.

\begin{thm}\label{thm:trivial def}
  Let $N$ be a Nijenhuis operator on a compatible Lie algebra $(\g,[-,-],\{-,-\})$. Then an infinitesimal
  deformation of  $(\g,[-,-],\{-,-\})$ can be
  obtained by putting
  \begin{eqnarray}
\omega_1(x,y)&=&[x, N(y)]+[N(x), y]-N[x, y],\\
\omega_2(x,y)&=&\{x, N(y)\}+\{N(x), y\}-N\{x, y\},
\end{eqnarray}
  for any $x,y\in\g$.   Furthermore, this infinitesimal
deformation is trivial.
\end{thm}
\begin{proof}
Since $(\omega_1,\omega_2)=\jetd^1 N$, $(\omega_1,\omega_2)$ is a $2$-cocycle. Since $N$ is a Nijenhuis operator on the compatible Lie algebra $(\g,[-,-],\{-,-\})$, $\omega_1(x,y)=[x,y]_N$ and $\omega_2(x,y)=\{x,y\}_N$, by Proposition \ref{pro:Nijenhuis operator property}, $(\g,\omega_1,\omega_2)$ is also a compatible Lie algebra. Thus \eqref{eq:2-closed omega bracket} holds and then $(\g,[-,-]_t,\{-,-\}_t)$ is an infinitesimal
deformation of  $(\g,[-,-],\{-,-\})$.

It is straightforward to verify that
\eqref{Nij1}-\eqref{Nij4} are satisfied. Thus this infinitesimal
deformation is trivial.
\end{proof}

\section{Cohomologies  of compatible Lie algebras: general theory}\label{sec:cohomology II}

In this section, we give   cohomologies of
compatible Lie algebras with coefficients in arbitrary
representations, and use the second cohomology group
to classify abelian extensions of a compatible Lie algebra. We also
introduce the  reduced cohomology of a compatible Lie algebra and
give its connection to Dubrovin and Zhang's cohomology of
compatible Poisson structures.

\subsection{Cohomology of a compatible Lie algebra  with coefficients in arbitrary representation}\label{sec:cohomology-gc}
In the following, we give the cohomology of a compatible Lie
algebra  associated to arbitrary representation.

Let $(\mathfrak{g},[-,-],\{-,-\})$ be a compatible Lie algebra with $\pi_1(x,y)=[x,y]$ and $\pi_2(x,y)=\{x,y\}$ and $(V;\rho,\mu)$  a representation of $(\mathfrak{g},[-,-],\{-,-\})$.

\begin{pro}
 With the above notations,  the pair $(\hat{\pi}_1+\hat{\rho},\hat{\pi}_2+\hat{\mu})$ is a Maurer-Cartan element of the bidifferential graded Lie algebra $(C^*(\g\oplus V,\g\oplus V),[-,-]_\NR,\partial_1=0,\partial_2=0)$, i.e.,
\begin{eqnarray}\label{eq:Maurer-Cartatn 12}
  [\hat{\pi}_1+\hat{\rho},\hat{\pi}_1+\hat{\rho}]_{\NR}=0,\quad[\hat{\pi}_2+\hat{\mu},\hat{\pi}_2+\hat{\mu}]_{\NR}=0,\quad [\hat{\pi}_1+\hat{\rho},\hat{\pi}_2+\hat{\mu}]_{\NR}=0.
\end{eqnarray}
\end{pro}
\begin{proof}
  Since $(\g,[-,-])$ is a Lie algebra and $(V;\rho)$ is a representation of $(\g,[-,-])$, by Proposition \ref{pro:MC}, we have $$[\hat{\pi}_1+\hat{\rho},\hat{\pi}_1+\hat{\rho}]_{\NR}=0.$$
  Similarly, since $(\g,\{-,-\})$ is a Lie algebra and $(V;\mu)$ is a representation of $(\g,\{-,-\})$, we have
  $$[\hat{\pi}_2+\hat{\mu},\hat{\pi}_2+\hat{\mu}]_{\NR}=0.$$
 For $x_1,x_2,x_3\in\g$ and $v_1,v_2,v_3\in V$, by \eqref{eq:cl} and \eqref{eq:r3} we have
  \begin{eqnarray*}
  &&[\hat{\pi}_1+\hat{\rho},\hat{\pi}_2+\hat{\mu}]_{\NR}((x_1,v_1),(x_2,v_2),(x_3,v_3))\\
  &=&(\hat{\pi}_1+\hat{\rho})\big((\hat{\pi}_2+\hat{\mu})((x_1,v_1),(x_2,v_2)),(x_3,v_3)\big)+(\hat{\pi}_1+\hat{\rho})\big((\hat{\pi}_2+\hat{\mu})((x_2,v_2),(x_3,v_3)),(x_1,v_1)\big)\\
   &&+(\hat{\pi}_1+\hat{\rho})\big((\hat{\pi}_2+\hat{\mu})((x_3,v_3),(x_1,v_1)),(x_2,v_2)\big)+(\hat{\pi}_2+\hat{\mu})\big((\hat{\pi}_1+\hat{\rho})((x_1,v_1),(x_2,v_2)),(x_3,v_3)\big)\\
   &&+(\hat{\pi}_2+\hat{\mu})\big((\hat{\pi}_1+\hat{\rho})((x_2,v_2),(x_3,v_3)),(x_1,v_1)\big)+(\hat{\pi}_2+\hat{\mu})\big((\hat{\pi}_1+\hat{\rho})((x_3,v_3),(x_1,v_1)),(x_2,v_2)\big)\\
   &=&\Big(\pi_1(\pi_2(x_1,x_2),x_3)+\pi_1(\pi_2(x_2,x_3),x_1)+\pi_1(\pi_2(x_3,x_1),x_2)\\
   &&+\pi_2(\pi_1(x_1,x_2),x_3)+\pi_2(\pi_1(x_2,x_3),x_1)+\pi_2(\pi_1(x_3,x_1),x_2),\\&&\rho(\pi_2(x_1,x_2))v_3+\mu(\pi_1(x_1,x_2))v_3-[\rho(x_1),\mu(x_2)](v_3)\\
   &&-[\rho(x_2),\mu(x_1)](v_3)+\rho(\pi_2(x_2,x_3))v_1+\mu(\pi_1(x_2,x_3))v_1\\
   &&-[\rho(x_2),\mu(x_3)](v_1)-[\rho(x_3),\mu(x_2)](v_1)+\rho(\pi_2(x_3,x_1))v_2\\
   &&+\mu(\pi_1(x_3,x_1))v_2-[\rho(x_3),\mu(x_1)](v_2)-[\rho(x_1),\mu(x_3)](v_2)\Big)\\
   &=&(0,0),
  \end{eqnarray*}
  which implies that $[\hat{\pi}_1+\hat{\rho},\hat{\pi}_2+\hat{\mu}]_{\NR}=0.$ Thus $(\hat{\pi}_1+\hat{\rho},\hat{\pi}_2+\hat{\mu})$ is a Maurer-Cartan element of the bidifferential graded Lie algebra $(C^*(\g\oplus V,\g\oplus V),[-,-]_\NR,\partial_1=0,\partial_2=0)$.
\end{proof}

By Proposition \ref{pro:CE-operator}, the Chevalley-Eilenberg  coboundary operators for the Lie algebra $(\g, [-,-])$ with coefficients in $(\g;\rho)$ and the Lie algebra $(\g, \{-,-\})$ with coefficients in $(\g;\mu)$ are given by
 $$\dM^n_{\pi_1+\rho}f=(-1)^{n-1}[\hat{\pi}_1+\hat{\rho},f]_\NR,\quad \dM^n_{\pi_2+\mu}f=(-1)^{n-1}[\hat{\pi}_2+\hat{\mu},f]_\NR,\quad \forall~f\in C^n(\g,V).$$
 By the graded Jacobi identity, \eqref{eq:Maurer-Cartatn 12} is also equivalent to
\begin{equation}\label{eq:MC-direct sum}
\dM^{n+1}_{\pi_1+\rho}\circ\dM^n_{\pi_1+\rho}=0,\quad \dM^{n+1}_{\pi_2+\mu}\circ\dM^n_{\pi_2+\mu}=0,\quad \dM^{n+1}_{\pi_1+\rho}\circ \dM^n_{\pi_2+\mu}+\dM^{n+1}_{\pi_2+\mu}\circ \dM^n_{\pi_1+\rho}=0.
\end{equation}

Define the space of $0$-cochains $\huaC^0(\g,V)$ by
$$\huaC^0(\g,V)=\{v\in V\mid \rho(x)v=\mu(x)v,~\forall~x\in\g\}.$$
For $n\ge 1$, define the space of $n$-cochains $\huaC^n(\g,V)$ by
$$\huaC^n(\g,V)=\underbrace{C^n(\g,V)\oplus C^n(\g,V)\cdots\oplus C^n(\g,V)}_{\mbox{$n$~{copies}}}.$$

Note that $\dM^0_{\pi_1+\rho}(v)(x)=\rho(x)v$ and $\dM^0_{\pi_2+\mu}(v)(x)=\mu(x)v$. Define $\delta^0:\huaC^{0}(\g,V)\longrightarrow \huaC^{1}(\g,V)$  by
$$\delta^0 v=\dM^0_{\pi_1+\rho} v=\dM^0_{\pi_2+\mu} v,\quad v\in \huaC^0(\g,V),$$
define $\delta^1:\huaC^{1}(\g,V)\longrightarrow \huaC^{2}(\g,V)$  by
$$\delta^1 f=(\dM^1_{\pi_1+\rho} f,\dM^1_{\pi_2+\mu} f),\quad f\in \huaC^1(\g,V)$$
and
define $\delta^n:\huaC^{n}(\g,V)\longrightarrow \huaC^{n+1}(\g,V)$  by
\begin{eqnarray*}
    \delta^n(\omega_1,\cdots,\omega_{n})=(\dM^n_{\pi_1+\rho}\omega_1,\cdots,\underbrace{\dM^n_{\pi_2+\mu}\omega_{i-1}+\dM^n_{\pi_1+\rho}\omega_i}_i,\cdots,\dM^n_{\pi_2+\mu}\omega_{n}),
\end{eqnarray*}
where $(\omega_1,\omega_2,\cdots,\omega_n)\in\huaC^n(\g,V)$ and $2\leq i\leq n$.

Similar to the proof of Theorem \ref{thm:cohomology of CLA} or using the relation \eqref{eq:MC-direct sum}, we also have $\delta^{n+1}\circ\delta^n=0$. Thus $(\huaC^*(\g,V)=\oplus_{n=0}^\infty \huaC^n(\g,V),\delta^*)$ is a cochain complex.

\begin{defi}Let $(\mathfrak{g},[-,-],\{-,-\})$ be a compatible Lie algebra and $(V;\rho,\mu)$  a
representation.
  The cohomology of the cochain complex $(\huaC^*(\g,V),\delta^*)$  is  called {\bf the cohomology of  $(\g,[-,-],\{-,-\})$ with coefficients in $V$}.  We denote the $n$-th cohomology group by  $\huaH^n(\g,V)$.
\end{defi}



\subsection{Abelian extensions of compatible Lie algebras}

In this subsection, we study abelian extensions of a compatible Lie
algebra and show that they are classified by the second
cohomology group.
\begin{defi}\label{defi:isomorphic}
\begin{itemize}
\item[\rm(1)] Let $(\g,[-,-]_\g,\{-,-\}_\g)$, $(\h,[-,-]_\h,\{-,-\}_\h)$, $(\hat{\g},[-,-]_{\hat{\g}},\{-,-\}_{\hat{\g}})$ be  compatible Lie algebras. An extension of compatible Lie algebras is a short exact sequence of compatible Lie algebras:
$$ 0\longrightarrow\h\stackrel{\id}{\longrightarrow}\hat{\g}\stackrel{\p}\longrightarrow\g\longrightarrow0.$$
We say that $\hat{\g}$ is an  {\bf extension} of $\g$ by $\h$. An
extension of $\g$ by $\h$ is called {\bf abelian} if the
compatible Lie algebra structure on $\h$ is trivial. \item[\rm(2)]
A {\bf linear section} of $\hat{\g}$ is a linear map
$\sigma:\g\rightarrow\hat{\g}$ such that $\p\circ \sigma=\Id$.
\item[\rm(3)] Two extensions $\hat{\g}_1$ and $\hat{\g}_2$ of $\g$
by $\h$   are said to be {\bf isomorphic} if there exists a
compatible Lie algebra   isomorphism
$\theta:\hat{\g}_1\longrightarrow \hat{\g}_2$ such that we have
the following commutative diagram:
\begin{equation}\label{diagram1}
\begin{array}{ccccccccc}
0&\longrightarrow& \h&\stackrel{\id_1}\longrightarrow&\hat{\g}_1&\stackrel{\p_1}\longrightarrow&\g&\longrightarrow&0\\
 &            &\Big\|&       &\theta\Big\downarrow&          &\Big\|& &\\
 0&\longrightarrow&\h&\stackrel{\id_2}\longrightarrow&\hat{\g}_2&\stackrel{\p_2}\longrightarrow&\g&\longrightarrow&0.
 \end{array}\end{equation}
\end{itemize}
\end{defi}

Let $\hat{\g}$ be an extension of $\g$ by $\h$, and
$\sigma:\g\rightarrow\hat{\g}$ a linear section. Define
$\omega_1,\omega_2:\g\otimes\g\rightarrow\h$ and
$\rho,\mu:\g\longrightarrow\gl(\h)$  respectively   by
\begin{eqnarray}
  \label{eq:str1}\omega_1(x,y)&=&[\sigma(x),\sigma(y)]_{\hat{\g}}-\sigma[x,y]_{\g},\\
  \label{eq:str2}\omega_2(x,y)&=&\{\sigma(x),\sigma(y)\}_{\hat{\g}}-\sigma\{x,y\}_{\g},\\
 \label{eq:str3} \rho(x)v&=&[\sigma(x),v]_{\hat{\g}},\\
 \label{eq:str4} \mu(x)v&=&\{\sigma(x),v\}_{\hat{\g}},\quad \forall x,y\in\g,~v\in \h.
\end{eqnarray}

Given a linear section, we have $\hat{\g}\cong\g\oplus \h$ as
vector spaces, and the compatible Lie algebra structure on
$\hat{\g}$ can be transferred to  $\g\oplus\h$:
\begin{eqnarray}
\label{eq:w-bracket1}[(x,u),(y,v)]_{(\omega_1,\rho)}&=&([x,y]_{\g},\rho(x)v-\rho(y)u+\omega_1(x,y)+[u,v]_\h),\\
\label{eq:w-bracket2}\{(x,u),(y,v)\}_{(\omega_2,\mu)}&=&(\{x,y\}_{\g},\mu(x)v-\mu(y)u+\omega_2(x,y)+\{u,v\}_\h).
\end{eqnarray}
Then $(\hat{\g}\cong\g\oplus
\h,[-,-]_{(\omega_1,\rho)},\{-,-\}_{(\omega_2,\mu)})$
is a compatible Lie algebra.

 Now we assume that the compatible Lie algebra structure on $\h$ is trivial, it is routine to check that  $(\h;\rho,\mu)$ is a representation of the compatible Lie algebra $(\g,[-,-]_\g,\{-,-\}_\g)$ and $(\omega_1,\omega_2)\in \huaC^{2}(\g,\h)$ is a $2$-cocycle of $\g$ with coefficients in the representation $(\h;\rho,\mu)$.

 If we choose another linear section $\sigma':\g\rightarrow \h$ and let $\omega_1',\omega_2'\in\Hom(\g\otimes \g ,\h)$ be the corresponding $2$-cocycle given by
$$\omega_1'(x,y)=[\sigma'(x),\sigma'(y)]_{\hat{\g}}-\sigma'([x,y]_\g),\quad \omega_2'(x,y)=\{\sigma'(x),\sigma'(y)\}_{\hat{\g}}-\sigma'(\{x,y\}_\g).$$

Set $\varphi=\sigma-\sigma'$. Since $\p\circ \sigma=\p\circ \sigma'=\Id$, we have
\begin{eqnarray*}
  \p\circ\varphi= \p\circ\sigma-\p\circ\sigma'=0,
\end{eqnarray*}
which implies that $\varphi\in \Hom(\g,\h)$.

Since $\h$ is abelian, thus
\begin{eqnarray*}
\rho(x)v&=&[\sigma(x), v]_{\hat{g}}=[\sigma'(x), v]_{\hat{g}},\\
\mu(x)v&=&\{\sigma(x),v\}_{\hat{g}}=\{\sigma'(x), v\}_{\hat{g}}.
\end{eqnarray*}
This means that the representation $(\h;\rho,\mu)$ of the compatible Lie algebra $\g$ does not depend on the choice of linear sections.
Then we obtain that
\begin{eqnarray*}
 (\omega_1-\omega_1')(x,y)&=&\rho(x)\varphi(y)-\rho(y)\varphi(x)-\varphi([x,y]_\g)=\dM^1_{\pi_1+\rho}\varphi(x,y),\\
 (\omega_2-\omega_2')(x,y)&=&\mu(x)\varphi(y)-\mu(y)\varphi(x)-\varphi(\{x,y\}_\g)=\dM^1_{\pi_2+\mu}\varphi(x,y).
\end{eqnarray*}
The classical argument in Eilenberg-MacLane cohomology
(\cite{Group extension}) may easily be extended to a proof of the
following   result.
\begin{thm}
Let $(\g,[-,-]_\g,\{-,-\}_\g)$ be a compatible Lie algebra. Assume
that the compatible Lie algebra structure on $\h$ is trivial and
$(\rho,\mu)$ is a representation of $\g$ on $\h$. Then the abelian
extension of $\g$ by $\h$  determined by a $2$-cocycle
$(\omega_1,\omega_2)\in \huaC^{2}(\g,\h) $ builds a one-to-one
correspondence between the   isomorphism
classes of abelian extensions of $\g$ by $\h$ and the cohomology
classes in $\huaH^2(\g,\h)$.
\end{thm}

\subsection{Reduced cohomologies  of compatible Lie algebras
and cohomologies of bi-Hamiltonian structures}  In this
subsection, we introduce the   reduced cohomology  of a compatible
Lie algebra and establish the relation between the reduced
cohomology of a compatible Lie algebra and the cohomology of the
corresponding compatible linear Poisson structures introduced by
Dubrovin and Zhang in \cite{DZ} in the study of
  bi-Hamiltonian structures.

Let $(\mathfrak{g},[-,-],\{-,-\})$ be a compatible Lie algebra and $(V;\rho,\mu)$ its representation.  Denote by $\tilde{C}^n(\g,V)=\ker~\dM^n_{\pi_1+\rho}\mid_{C^n(\g,V)}$. By the fact $\dM^{n+1}_{\pi_1+\rho}\circ \dM^{n}_{\pi_2+\mu}+\dM^{n+1}_{\pi_2+\mu}\circ \dM^n_{\pi_1+\rho}=0$,  for $f\in \tilde{C}^n(\g,V)$, we have
$$\dM^{n+1}_{\pi_1+\rho} \dM^n_{\pi_2+\mu}(f)=-\dM^{n+1}_{\pi_2+\mu}\circ \dM^n_{\pi_1+\rho}(f)=0,$$
which implies that $\dM^n_{\pi_2+\mu} f\in \tilde{C}^{n+1}(\g,V)$.
Since $\dM^{n+1}_{\pi_2+\mu}\circ\dM_{\pi_2+\mu}^n=0$,  $(\tilde{C}^\ast(\g,V)=\oplus_{n=0}^\infty
\tilde{C}^n(\g,V),\dM^\ast_{\pi_2+\mu})$ is a cochain complex.

\begin{defi}  Let $(\mathfrak{g},[-,-],\{-,-\})$ be a compatible Lie algebra and $(V;\rho,\mu)$ a
representation.  The cohomology of the cochain complex
$(\tilde{C}^*(\g,V),\dM^*_{\pi_2+\mu})$  is  called {\bf the
reduced cohomology} of  $(\g,[-,-],\{-,-\})$ with
coefficients in $V$. We denote the $n$-th cohomology group by
$\tilde{H}^n(\g,V)$.
\end{defi}

\emptycomment{The $2$-th cohomology group is given by
$$\tilde{H}^2(\g,V)=\frac{\ker~\dM^2_\mu\mid_{\tilde{C}^2(\g,V)}}{\Img~\dM^1_\mu\mid_{\tilde{C}^1(\g,V)}},$$
where $\ker~\dM^2_\mu\mid_{\tilde{C}^2(\g,V)}$ and $\Img~\dM^1_\mu\mid_{\tilde{C}^1(\g,V)}$ are evidently given by
\begin{eqnarray*}
    \ker~\dM^2_\mu\mid_{\tilde{C}^2(\g,V)}&=&\{\omega\in C^{2}(\g,V)\mid \dM^2_\rho \omega=\dM^2_\mu\omega=0\},\\
    \Img~\dM^1_\mu\mid_{\tilde{C}^1(\g,V)}&=&\{\omega\in C^{2}(\g,V)\mid \exists~\mbox{ $f\in \Hom(\g,V)$ ~s.t. }~ \dM^1_\rho f=0,~\omega=\dM^1_\mu f \}.
    \end{eqnarray*}}

Let $M$ be a smooth manifold. The space of multi-vector fields on $M$: $$\frkX^\ast(M)=\oplus_{n=0}^{+\infty}\frkX^n(M), \mbox{ with } \frkX^n(M):=\Gamma(\wedge^n TM)$$
 carries a Schouten bracket $[-,-]_S:\frkX^p(M)\otimes\frkX^q(M)\rightarrow \frkX^{p+q+1}(M) $, which makes $\frkX^\ast(M)$ into a graded Lie algebra. A Poisson structure on $M$ is a bivector field $\Pi\in \frkX^2(M)$ satisfying
$$[\Pi,\Pi]_S=0.$$
The Poisson coboundary operator $\delta^p_\Pi:\frkX^p(M)\rightarrow \frkX^{p+1}(M)$ is given by
$$\delta^p_\Pi(P)=(-1)^{p-1}[\Pi,P]_S,\quad\forall~P\in\frkX^p(M).$$
The cohomology of the cochain complex $(\frkX^\ast(M),\delta^*_\Pi)$ is called the Poisson cohomology complex of $M$.  We denote the corresponding $n$-th cohomology group by  $H^n(M,\Pi)$.

Assume that the Poisson manifold $(M,\Pi_1)$ is endowed with a second Poisson structure $\Pi_2$ which is compatible with $\Pi_1$ in the sense that any linear combination of $\Pi_1$ and $\Pi_2$ is still a Poisson structure on $M$. In this case we say that the manifold has a bi-Hamiltonian structure $(\Pi_1,\Pi_2)$. The two Poisson structures define on $\frkX^\ast(M)$ two coboundary operators:
$$\delta^p_{\Pi_1}(P)=(-1)^{p-1}[\Pi_1,P]_S,\quad \delta^p_{\Pi_2}(P)=(-1)^{p-1}[\Pi_2,P]_S,\quad\forall~P\in\frkX^p(M),$$
which satisfies
$$\delta_{\Pi_1}^{p+1}\circ \delta^p_{\Pi_2}+\delta^{p+1}_{\Pi_2}\circ \delta^p_{\Pi_1}=0.$$
Denote by $\tilde{\frkX}^p(M)=\ker~\delta^p_{\Pi_1}\mid_{\frkX^p(M)}$. Then $(\tilde{\frkX}^\ast(M)=\oplus_{n=0}^\infty \tilde{\frkX}^n(M),\delta^*_{\Pi_2})$ is a cochain complex, which is called the {\bf bi-Hamiltonian cohomology complex} of $M$ given by Dubrovin and Zhang in \cite{DZ}. See \cite{DLZ,DZ,LZ} for more details on this cohomology complex and its applications.  We denote the corresponding $n$-th cohomology group by  $\tilde{H}^n(M,\Pi_1,\Pi_2)$.

Recall that a Poisson structure on a vector space is called {\bf linear} if the set of linear
functions is closed under the Poisson bracket. Usually the linear Poisson structures are also called Lie-Poisson structures. The Lie-Poisson structures are in one-to-one correspondence with Lie algebra structures. Namely, let $(\g,[-,-])$ be a finite dimensional Lie algebra and the associated Lie-Poisson structure $\Pi$ on $\g^*$ is determined by
$$\Pi(d l_x,d l_y)=\{l_x,l_y\}=l_{[x,y]},\quad \forall~x,y\in\g,$$
where for $x\in\g$, $l_x$ is the linear function on $\g^*$ defined by
$$l_x(\xi)=\langle x,\xi\rangle,\quad\forall~\xi\in \g^*.$$
Furthermore, the linear map $\rho:\g\rightarrow\gl(C^{\infty}(\g^*))$ defined by
$$\rho(x)(f)=\{x,f\},\quad \forall x\in\g, f\in C^{\infty}(\g^*)$$
gives a representation of the Lie algebra $\g$ on $C^{\infty}(\g^*)$. Then the Poisson cohomology complex $(\frkX^\ast(\g^*),\delta^*_\Pi)$ is isomorphic to the Chevalley-Eilenberg complex
of $\g$ with coefficients in $C^{\infty}(\g^*)$ (\cite{LPV})
$$(\frkX^\ast(\g^*),\delta^*_\Pi)\backsimeq (C^*(\g,C^{\infty}(\g^*)),\dM^*_{\pi+\rho}).$$

Compatible Lie-Poisson structures on the  vector space $\g^*$ are also in one-to-one correspondence with compatible Lie algebra structures on the vector space $\g$. Namely, let $(\g,[-,-]_1,[-,-]_2)$ be a compatible Lie algebra. Let $\Pi_1$ and $\Pi_2$ be the corresponding Lie-Poisson structures for the Lie algebras $(\g,[-,-]_1)$ and $(\g,[-,-]_2)$, respectively. Then $(\Pi_1,\Pi_2)$ is a bi-Hamiltonian structure on the vector space $\g^*$. Define linear maps
 \begin{eqnarray}\label{eq:compatible-rep}
 \rho(x)(f)=\{x,f\}_1,\quad \mu(x)(f)=\{x,f\}_2,\quad \forall x\in\g, f\in C^{\infty}(\g^*),
 \end{eqnarray}
 where $\{-,-\}_1$ and $\{-,-\}_2$ are Poisson brackets of the Poisson structures $\Pi_1$ and $\Pi_2$, respectively. Then we have
 \begin{pro}\label{pro:new representation}
   Let $(\g,[-,-]_1,[-,-]_2)$ be a compatible Lie algebra. Then the pair $(\rho,\mu)$ defined by \eqref{eq:compatible-rep} gives a representation of the compatible Lie algebra $\g$ on the vector space $C^\infty(\g^*)$.
 \end{pro}
 \begin{proof}
 Since $(\g,[-,-]_1,[-,-]_2)$ is a compatible Lie algebra, $(C^\infty(\g^*),\{-,-\}_1,\{-,-\}_2)$ is a compatible Poisson algebra. The Jacobi identity of the Poisson algebra $(C^\infty(\g^*),\{-,-\}_1)$ implies that $\rho$ is a representation of the Lie algebra $(\g,[-,-]_1)$ on $C^\infty(\g^*)$,  the Jacobi identity of the Poisson algebra $(C^\infty(\g^*),\{-,-\}_2)$ implies that $\mu$ is a representation of the Lie algebra $(\g,[-,-]_2)$ on $C^\infty(\g^*)$ and the compatibility condition between Poisson algebras $(C^\infty(\g^*),\{-,-\}_1)$ and $(C^\infty(\g^*),\{-,-\}_2)$  implies that \eqref{eq:r3} holds. Thus $(\rho,\mu)$ is a representation of the compatible Lie algebra $\g$ on the vector space $C^\infty(\g^*)$.
 \end{proof}
 By Proposition \ref{pro:new representation}, we have a reduced cohomology complex $(\tilde{C}^*(\g,C^{\infty}(\g^*)),\dM^*_{\pi_2+\mu})$ associated to the representation $(C^\infty(\g^*);\rho,\mu)$ of the compatible Lie algebra $(\g,[-,-]_1,[-,-]_2)$. Furthermore, the following theorem says that the reduced cohomology complex of a compatible Lie algebra is isomorphic to the corresponding bi-Hamiltonian cohomology complex.
\begin{thm}
  Let $(\g,[-,-]_1,[-,-]_2)$ be a compatible Lie algebra and $\Pi_1,\Pi_2$  the corresponding compatible Lie-Poisson structures on $\g^*$.
Then the cohomology complex
  $(\tilde{\frkX}^\ast(\g^*),\delta^*_{\Pi_2})$ of the bi-Hamiltonian structure $(\Pi_1,\Pi_2)$ on $\g^*$ is isomorphic to the reduced cohomology complex $(\tilde{C}^*(\g,C^{\infty}(\g^*)),\dM^*_{\pi_2+\mu})$ of the compatible Lie algebra $\g$. Therefore, for any $n\in\Nat$, there is a canonical isomorphism
  $$\tilde{H}^n(\g^*,\Pi_1,\Pi_2)\simeq \tilde{H}^n(\g,C^{\infty}(\g^*)).$$
\end{thm}
\begin{proof}
  Let $P$ be a Poisson $p$-cochain of the Poisson structure $\Pi_2$, i.e. $P\in \frkX^p(\g^*)$. The element $\Psi_p(P)\in C^p(\g,C^{\infty}(\g^*))$ associated to $P$ is defined by
  $$\Psi_p(P)(x_1,\cdots,x_p)=P(dl_{x_1},\cdots,dl_{x_p}).$$
  Then for any $p\in\Nat$, the linear map $\Psi_p:\frkX^p(\g^*)\rightarrow C^p(\g,C^{\infty}(\g^*))$ is an isomorphism and satisfies the following relations (\cite{LPV})
  \begin{eqnarray}
 \label{eq:commu1} \Psi_{p+1}\circ \delta^p_{\Pi_1}(P)&=&\dM^{p}_{\pi_1+\rho}\circ\Psi_{p}(P),\\
 \label{eq:commu2} \Psi_{p+1}\circ \delta^p_{\Pi_2}(P)&=&\dM^p_{\pi_2+\mu}\circ\Psi_{p}(P).
  \end{eqnarray}
  Assume that $P\in \tilde{\frkX}^p(\g^*)$, i.e. $\delta^p_{\Pi_1}P=0$. By \eqref{eq:commu1}, $\Psi_{p}(P)\in \tilde{C}^p(\g,C^{\infty}(\g^*))$. Thus for any $p\in\Nat$, the linear map $\Psi_p:\tilde{\frkX}^p(\g^*)\rightarrow \tilde{C}^p(\g,C^{\infty}(\g^*))$ is well-defined and furthermore, the linear map $\Psi_p$ is an isomorphism. By \eqref{eq:commu2}, for any $P\in \tilde{\frkX}^p(\g^*)$, we have
  $$\Psi_{p+1}\circ \delta^p_{\Pi_2}(P)=\dM^p_{\pi_2+\mu}\circ\Psi_{p}(P).$$
  Thus the cohomology complex
  $(\tilde{\frkX}^\ast(\g^*),\delta^*_{\Pi_2})$ of the bi-Hamiltonian structure $(\Pi_1,\Pi_2)$ on $\g^*$ is isomorphic to the reduced cohomology complex $(\tilde{C}^*(\g,C^{\infty}(\g^*)),\dM^*_{\pi_2+\mu})$ of the compatible Lie algebra $\g$. The isomorphisms $(\Psi_p)_{p\in N}$ therefore induce isomorphisms between the corresponding cohomology groups.
\end{proof}

\section{Nonabelian extensions of compatible Lie algebras}\label{sec:cohomology III}
In this section, we use the Maurer-Cartan approach to classify
nonabelian extensions of compatible Lie algebras.

Let
$(\g,[-,-]_\g,\{-,-\}_\g)$ and $(\h,[-,-]_\h,\{-,-\}_\h)$ be two
compatible Lie algebras. Set
$$\pi_1(x,y)=[x,y]_\g,\quad \pi_2(x,y)=\{x,y\}_\g,\quad\vartheta_1(u,v)=[u,v]_\h,\quad \vartheta_2(u,v)=\{u,v\}_\h.$$
Let $\hat{\g}$ be a nonabelian extension of the compatible Lie algebra $\g$ by $\h$, i.e. the compatible Lie algebra structure on $\h$ is not trivial. Given a linear section $\sigma:\g\rightarrow\hat{\g}$, the compatible Lie algebra structure on $\hat{\g}$ is isomorphic  to the compatible Lie algebra structure $(\hat{\g}=\g\oplus\h,[-,-]_{(\omega_1,\rho)},\{-,-\}_{(\omega_2,\mu)})$, where $[-,-]_{(\omega_1,\rho)}$ and $\{-,-\}_{(\omega_2,\mu)}$ are given by \eqref{eq:w-bracket1} and  \eqref{eq:w-bracket2} respectively.
Then we have
\begin{pro}
Let $(\g,[-,-]_\g,\{-,-\}_\g)$ and $(\h,[-,-]_\h,\{-,-\}_\h)$ be two  compatible Lie algebras. $(\g\oplus\h,[-,-]_{(\omega_1,\rho)},\{-,-\}_{(\omega_2,\mu)})$ is a compatible Lie algebra  if and only if
\begin{equation}\label{eq:extension1}
 \begin{array}{rclrclrcl}
[\hat{\pi}_1+\hat{\rho},\hat{\pi}_1+\hat{\rho}]_{\NR}+2[\hat{\omega}_1,\hat{\vartheta}_1]_{\NR}&=&0,&[\hat{\pi}_1+\hat{\rho},\hat{\omega}_1]_{\NR}&=&0,& [\hat{\rho},\hat{\vartheta}_1]_{\NR}&=&0,\\
{[\hat{\pi}_2+\hat{\mu},\hat{\pi}_2+\hat{\mu}]}_{\NR}+2[\hat{\omega}_2,\hat{\vartheta}_2]_{\NR}&=&0,&[\hat{\pi}_2+\hat{\mu},\hat{\omega}_2]_{\NR}&=&0,& [\hat{\mu},\hat{\vartheta}_2]_{\NR}&=&0
 \end{array}
\end{equation}
and
\begin{equation}\label{eq:extension2}
 \begin{array}{rclrcl}
{[\hat{\rho},\hat{\vartheta}_1]}_{\NR}+[\hat{\mu},\hat{\vartheta}_2]_{\NR}&=&0, & & &\\
{[\hat{\pi}_1+\hat{\rho},\hat{\omega}_1]}_{\NR}+[\hat{\pi}_2+\hat{\mu},\hat{\omega}_2]_{\NR}&=&0,& & &\\
{[\hat{\pi}_1+\hat{\rho},\hat{\pi}_2+\hat{\mu}]}_{\NR}+[\hat{\omega}_1,\hat{\vartheta}_2]_{\NR}+[\hat{\omega}_2,\hat{\vartheta}_1]_{\NR}&=&0.& & &
 \end{array}
\end{equation}
\end{pro}
\begin{proof}
Note that $(\g\oplus\h,[-,-]_{(\omega_1,\rho)},\{-,-\}_{(\omega_2,\mu)})$ is a compatible Lie algebra if and only if the pair $(\hat{\pi}_1+\hat{\rho}+\hat{\omega}_1+\hat{\vartheta}_1,\hat{\pi}_2+\hat{\mu}+\hat{\omega}_2+\hat{\vartheta}_2)$ is a Maurer-Cartan element of the bidifferential graded Lie algebra $(C^*(\g\oplus V,\g\oplus V),[-,-]_\NR,\partial_1=0,\partial_2=0)$, i.e.,
  \begin{eqnarray}
 \label{eq:EXT1}   [\hat{\pi}_1+\hat{\rho}+\hat{\omega}_1+\hat{\vartheta}_1,\hat{\pi}_1+\hat{\rho}+\hat{\omega}_1+\hat{\vartheta}_1]_{\NR}&=&0,\\
 \label{eq:EXT2}     [\hat{\pi}_2+\hat{\mu}+\hat{\omega}_2+\hat{\vartheta}_2,\hat{\pi}_2+\hat{\mu}+\hat{\omega}_2+\hat{\vartheta}_2]_{\NR}&=&0,\\
   \label{eq:EXT3}   [\hat{\pi}_1+\hat{\rho}+\hat{\omega}_1+\hat{\vartheta}_1,\hat{\pi}_2+\hat{\mu}+\hat{\omega}_2+\hat{\vartheta}_2]_{\NR}&=&0.
  \end{eqnarray}
 By the definition of bidegrees, we have $||\hat{\pi}_1||=||\hat{\pi}_2||=||\hat{\rho}||=||\hat{\mu}||=1|0$, $||\hat{\omega}_1||=||\hat{\omega}_2||=2\mid -1$ and $||\hat{\vartheta}_1||=||\hat{\vartheta}_2||=0|1$. By Lemma \ref{lem:bidegree perserve}, \eqref{eq:EXT1} and \eqref{eq:EXT2} imply that \eqref{eq:extension1} holds and \eqref{eq:EXT3} implies that \eqref{eq:extension2} holds.
\end{proof}

Moreover, we have
\begin{lem}\label{lem:nonabelian relations}
  Let $(\g,[-,-]_\g,\{-,-\}_\g)$ and $(\h,[-,-]_\h,\{-,-\}_\h)$ be two  compatible Lie algebras. Then $(\g\oplus\h,[-,-]_{(\omega_1,\rho)},\{-,-\}_{(\omega_2,\mu)})$ is a compatible Lie algebra if and only if
  \begin{eqnarray}
  \label{eq:EXT-1}  \rho([x,y]_\g)&=&[\rho(x),\rho(y)]+\ad^{\h}_{\omega_1(x,y)},\\
   \label{eq:EXT-2} \mu(\{x,y\}_\g)&=&[\mu(x),\mu(y)]+\AD^{\h}_{\omega_2(x,y)},\\
   \label{eq:EXT-3} \rho(x)[u,v]_\h&=&[\rho(x)u,v]_\h+[u,\rho(x)v]_\h,\\
    \label{eq:EXT-4}\mu(x)\{u,v\}_\h&=&\{\mu(x)u,v\}_\h+\{u,\mu(x)v\}_\h,\\
     \label{eq:EXT-5}\rho(\{x,y\}_\g)+\mu([x,y]_\g)&=&[\rho(x),\mu(y)]+[\mu(x),\rho(y)]-\ad^{\h}_{\omega_2(x,y)}-\AD^{\h}_{\omega_1(x,y)},\\
    \label{eq:EXT-6} \rho(x)\{u,v\}_\h+\mu(x)[u,v]_\h&=&\{\rho(x)u,v\}_\h+\{u,\rho(x)v\}_\h+[\mu(x)u,v]_\h+[u,\mu(x)v]_\h,\\
    \label{eq:EXT-7}  \rho(x)\omega_1(y,z)+\rho(z)\omega_1(x,y)&=&\rho(y)\omega_1(x,z)+\omega_1([x,y]_\g,z)+\omega_1([z,x]_\g,y)\\
    \nonumber&&+\omega_1([y,z]_\g,x),\\
    \label{eq:EXT-8} \mu(x)\omega_2(y,z)+\mu(z)\omega_2(x,y)&=&\mu(y)\omega_2(x,z)+\omega_2(\{x,y\}_\g,z)+\omega_1(\{z,x\}_\g,y)\\
     \nonumber&&+\omega_2(\{y,z\}_\g,x)
  \end{eqnarray}
  and
  \begin{eqnarray}
   \label{eq:EXT-9} &&\rho(x)\omega_2(y,z)+\rho(z)\omega_2(x,y)+\rho(y)\omega_2(z,x)+ \mu(x)\omega_1(y,z)+\mu(z)\omega_1(x,y)+\mu(y)\omega_1(z,x)\\
  \nonumber  &&=\omega_2([x,y]_\g,z)+\omega_2([z,x]_\g,y)+\omega_2([y,z]_\g,x)+\omega_1(\{x,y\}_\g,z)+\omega_1(\{z,x\}_\g,y)+\omega_1(\{y,z\}_\g,x),
  \end{eqnarray}
  where $\ad^\h_u v=[u,v]_\h$ and $\AD^\h _uv=\{u,v\}_\h$ for $u,v\in\h$.
\end{lem}
\begin{proof}
  Assume that $(\g\oplus\h,[-,-]_{(\omega_1,\rho)},\{-,-\}_{(\omega_2,\mu)})$ is a compatible Lie algebra. The equations

  $$([\hat{\pi}_1+\hat{\rho}+\hat{\omega}_1+\hat{\vartheta}_1,\hat{\pi}_1+\hat{\rho}+\hat{\omega}_1+\hat{\vartheta}_1]_{\NR})((x,0),(y,0),(0,u))=0$$ and  ${[\hat{\pi}_1,\hat{\pi}_1]}_{\NR}=0$ imply that \eqref{eq:EXT-1} holds.

  The equations $$([\hat{\pi}_2+\hat{\mu}+\hat{\omega}_2+\hat{\vartheta}_2,\hat{\pi}_2+\hat{\mu}+\hat{\omega}_2+\hat{\vartheta}_2]_{\NR})((x,0),(y,0),(0,u))=0$$ and ${[\hat{\pi}_2,\hat{\pi}_2]}_{\NR}=0$ imply that \eqref{eq:EXT-2} holds.

  The equations
  $$[\hat{\rho},\hat{\vartheta}_1]_{\NR}((x,0),(0,u),(0,v))=0,\quad [\hat{\mu},\hat{\vartheta}_2]_{\NR}((x,0),(0,u),(0,v))=0$$
  imply that \eqref{eq:EXT-3} and \eqref{eq:EXT-4} hold, respectively.

  The equations
  $$[\hat{\pi}_1+\hat{\rho},\hat{\omega}_1]_{\NR}((x,0),(y,0),(z,0))=0,\quad [\hat{\pi}_2+\hat{\mu},\hat{\omega}_2]_{\NR}((x,0),(y,0),(z,0))=0$$
  imply that \eqref{eq:EXT-7} and \eqref{eq:EXT-8} hold, respectively.

  The equations
  $$\big({[\hat{\pi}_1+\hat{\rho},\hat{\pi}_2+\hat{\mu}]}_{\NR}+[\hat{\omega}_1,\hat{\vartheta}_2]_{\NR}+[\hat{\omega}_2,\hat{\vartheta}_1]_{\NR}\big)((x,0),(0,u),(0,v))=0$$
  and ${[\hat{\pi}_1,\hat{\pi}_2]}_{\NR}=0$ imply that \eqref{eq:EXT-5} holds.

  The equation
  $$\big({[\hat{\rho},\hat{\vartheta}_1]}_{\NR}+[\hat{\mu},\hat{\vartheta}_2]_{\NR}\big)((x,0),(0,u),(0,v))=0$$
  implies that \eqref{eq:EXT-6} holds.

  The equation
$$\big({[\hat{\pi}_1+\hat{\rho},\hat{\omega}_1]}_{\NR}+[\hat{\pi}_2+\hat{\mu},\hat{\omega}_2]_{\NR}\big)((x,0),(y,0),(z,0))=0$$
implies that \eqref{eq:EXT-9} holds.

Conversely, if \eqref{eq:EXT-1}-\eqref{eq:EXT-9} hold, it is straightforward to check that $(\g\oplus\h,[-,-]_{(\omega_1,\rho)},\{-,-\}_{(\omega_2,\mu)})$ is a compatible Lie algebra.
\end{proof}

\begin{pro}\label{pro:isomorphism}
  Let $(\g,[-,-]_\g,\{-,-\}_\g)$ and $(\h,[-,-]_\h,\{-,-\}_\h)$ be two  compatible Lie algebras. Assume that $(\g\oplus\h,[-,-]_{(\omega_1,\rho)},\{-,-\}_{(\omega_2,\mu)})$ and $(\g\oplus\h,[-,-]_{(\omega_1',\rho')},\{-,-\}_{(\omega_2',\mu')})$ are two extensions of $\g$ by $\h$. Then they are isomorphic if and only if there exists a linear map $\xi:\g\rightarrow \h$ such that the following equalities holds:
\begin{eqnarray}
 \label{eq:isomorphism1} \rho'(x)-\rho(x)&=&\ad^\h_{\xi(x)},\\
 \label{eq:isomorphism2} \mu'(x)-\mu(x)&=&\AD^\h_{\xi(x)},\\
 \label{eq:isomorphism3} \omega_1'(x,y)-\omega_1(x,y)&=&\rho(x)\xi(y)-\rho(y)\xi(x)-\xi([x,y]_\g)+[\xi(x),\xi(y)]_\h,\\
 \label{eq:isomorphism4} \omega_2'(x,y)-\omega_2(x,y)&=&\mu(x)\xi(y)-\mu(y)\xi(x)-\xi(\{x,y\}_\g)+\{\xi(x),\xi(y)\}_\h,\quad\forall~x,y\in\g.
\end{eqnarray}
\end{pro}
\begin{proof}
  Let $(\g\oplus\h,[-,-]_{(\omega_1,\rho)},\{-,-\}_{(\omega_2,\mu)})$ and $(\g\oplus\h,[-,-]_{(\omega_1',\rho')},\{-,-\}_{(\omega_2',\mu')})$ be two extensions of $\g$ by $\h$. Assume that they are isomorphic. Then there exists a compatible Lie algebra isomorphism $\theta:\g\oplus\h\longrightarrow\g\oplus\h$, such that the following  diagram is commutative:
 \begin{equation*}
\begin{array}{ccccccccc}
0&\longrightarrow&\h&\stackrel{\id_1}\longrightarrow&\g\oplus\h_{(\omega_1,\omega_2,\rho,\mu)}&\stackrel{\p_1}\longrightarrow&\g&\longrightarrow&0\\
 &            &\Big\|&       &\theta\Big\downarrow&          &\Big\|& &\\
 0&\longrightarrow&\h&\stackrel{\id_2}\longrightarrow&\g\oplus\h_{(\omega'_1,\omega'_2,\rho',\mu')}&\stackrel{\p_2}\longrightarrow&\g&\longrightarrow&0.
 \end{array}\end{equation*}
Since $\p_2\circ\theta=\p_1$, we can assume that $\theta(x,u)=(x,-\xi(x)+u)$ for some $\xi:\g\longrightarrow\h$.

  Since
$\theta[(x,0),(0,v)]_{(\omega_1,\rho)}=[\theta(x,0),\theta(0,v)]_{(\omega'_1,\rho')}$
and
$\theta\{(x,0),(0,v)\}_{(\omega_2,\mu)}=\{\theta(x,0),\theta(0,v)\}_{(\omega'_2,\mu')}$,
we   deduce that \eqref{eq:isomorphism1}
and \eqref{eq:isomorphism2} hold.

 Since
$\theta[(x,0),(y,0)]_{(\omega_1,\rho)}=[\theta(x,0),\theta(y,0)]_{(\omega'_1,\rho')}$
and
$\theta\{(x,0),(y,0)\}_{(\omega_2,\mu)}=\{\theta(x,0),\theta(y,0)\}_{(\omega'_2,\mu')}$,
we   deduce that \eqref{eq:isomorphism3}
and \eqref{eq:isomorphism4} hold.

The converse can be proved similarly. We omit the details.
\end{proof}

Let $(\g,\pi_1=[-,-]_\g,\pi_2=\{-,-\}_\g)$ and $(\h,\vartheta_1=[-,-]_\h,\vartheta_2=\{-,-\}_\h)$ be two  compatible Lie algebras. Then $(\g\oplus\h,\pi_1+\vartheta_1,\pi_2+\vartheta_2)$ is a compatible Lie algebra.
By Theorem \ref{pro:new differential Lie algebra}, there is a bidifferential graded Lie algebra $(C^*(\g\oplus \h,\g\oplus\h),[-,-]_{\NR},\partial_1,\partial_2)$, where $\partial_i$ for $i=1,2$ is defined by
$$\partial_i P={[\hat{\pi}_i+\hat{\vartheta}_i,P]}_{\NR},\quad\forall~P\in C^p(\g\oplus\h,\g\oplus\h).$$
Define $C^k_>(\g\oplus\h,\h)\subset C^k(\g\oplus\h,\h)$ by
$$C^k(\g\oplus\h,\h)={C}^k_>(\g\oplus\h,\h)\oplus C^k(\h,\h).$$
Denote by $C^*_>(\g\oplus\h,\h)=\oplus_kC^k_>(\g\oplus\h,\h)$, which is a graded vector space.

\begin{lem}\label{lem:dgla}
  With the above notations, $({C}^*_>(\g\oplus\h,\h),[-,-]_{\NR},\partial_1,\partial_2)$ is a bidifferential graded subalgebra of $(C^*(\g\oplus\h,\g\oplus\h), [-,-]_{\NR},\partial_1,\partial_2)$.    Furthermore, its degree $0$ part $C^0_>(\g\oplus\h,\h)=\Hom(\g,\h)$ is abelian.
\end{lem}

\begin{proof}
  By the definition of $[-,-]_{\NR}$ and $k_1\partial_1+k_2\partial_2=[k_1(\hat{\pi}_1+\hat{\vartheta}_1)+k_2(\hat{\pi}_2+\hat{\vartheta}_2),-]_{\NR}$, we can show that $({C}^*(\g\oplus\h,\h),[-,-]_{\NR},\partial_1,\partial_2)$ is a bidifferential graded subalgebra of $(C^*(\g\oplus\h,\g\oplus\h), [-,-]_{\NR},\partial_1,\partial_2)$. An element  $P\in C_>^p(\g\oplus\h,\h)$ can be regarded as an element in $C^p(\g\oplus\h,\h)$ such that $P\mid_{C^p(\h,\h)}=0$. Furthermore, for $P\in C_>^p(\g\oplus\h,\h)$ and $Q\in C_>^q(\g\oplus\h,\h)$, it is straightforward to check that $[P,Q]_{\NR}|_{C^{p+q-1}(\h,\h)}=0$ and $(k_1\partial_1+k_2\partial_2) P|_{C^{p+1}(\h,\h)}=0$. Thus $({C}^*_>(\g\oplus\h,\h),[-,-]_{\NR},\partial_1,\partial_2)$ is a bidifferential graded subalgebra of $(C^*(\g\oplus\h,\g\oplus\h), [-,-]_{\NR},\partial_1,\partial_2)$. Obviously, $C^0_>(\g\oplus\h,\h)=\Hom(\g,\h)$ is abelian.
\end{proof}

We  denote this bidifferential graded Lie algebra  by $(L=\oplus_k
L_k,[-,-]_{\NR},\partial_1,\partial_2)$, where
$L_k=C^k_>(\g\oplus\h,\h).$

\begin{pro}\label{pro:JMC}
Let $(\g,[-,-]_\g,\{-,-\}_\g)$ and $(\h,[-,-]_\h,\{-,-\}_\h)$ be two  compatible Lie algebras. The following two statements are equivalent:
\begin{itemize}
  \item[\rm(i)] $(\g\oplus\h,[-,-]_{(\omega_1,\rho)},\{-,-\}_{(\omega_2,\mu)})$ is a compatible Lie algebra, which is a nonabelian extension of $\g$ by $\h$.

  \item[\rm(ii)]The pair $(\hat{\rho}+\hat{\omega}_1,\hat{\mu}+\hat{\omega}_2)$ is a Maurer-Cartan element of the bidifferential graded Lie algebra $(L=\oplus_k L_k,[-,-]_{\NR},\partial_1,\partial_2)$.
\end{itemize}
Thus there is a one-to-one correspondence between  nonabelian extensions of the compatible Lie algebra $\g$ by $\h$ and Maurer-Cartan elements of  the bidifferential graded Lie algebra $(L=\oplus_k L_k,[-,-]_{\NR},\partial_1,\partial_2)$.
\end{pro}
\begin{proof}
By Lemma \ref{lem:nonabelian relations}, $(\g\oplus\h,[-,-]_{(\omega_1,\rho)},\{-,-\}_{(\omega_2,\mu)})$ is a compatible Lie algebra if and only if \eqref{eq:EXT-1}-\eqref{eq:EXT-9} hold.

It was shown in \cite{nonabelin cohomology of Lie} that $\hat{\rho}+\hat{\omega}_1$ is a Maurer-Cartan element of the differential graded Lie algebra $(L=\oplus_k L_k,[-,-]_{\NR},\partial_1)$ if and only if \eqref{eq:EXT-1}, \eqref{eq:EXT-3} and \eqref{eq:EXT-7} hold and $\hat{\mu}+\hat{\omega}_2$ is a Maurer-Cartan element of the differential graded Lie algebra $(L=\oplus_k L_k,[-,-]_{\NR},\partial_2)$ if and only if \eqref{eq:EXT-2}, \eqref{eq:EXT-4} and \eqref{eq:EXT-8} hold.

Furthermore, for $e_i=(x_i,v_i)\in\g\oplus\h$, by a direct calculation, we have
\begin{eqnarray*}
  &&\partial_1 (\hat{\mu}+\hat{\omega}_2)(e_1,e_2,e_3)\\
  &=&-[v_1,\omega_2(x_2,x_3)]_\h-[v_2,\omega_2(x_3,x_1)]_\h-[v_3,\omega_2(x_1,x_2)]_\h+\omega_2([x_1,x_2]_\g,x_3)+\omega_2([x_2,x_3]_\g,x_1)\\
  &&+\omega_2([x_3,x_1]_\g,x_2)-[v_1,\mu(x_2)v_3]_\h-[v_2,\mu(x_3)v_1]_\h-[v_3,\mu(x_1)v_2]_\h+[v_1,\mu(x_3)v_2]_\h\\
  &&+[v_2,\mu(x_1)v_3]_\h+[v_3,\mu(x_2)v_1]_\h+\mu([x_1,x_2]_\g)v_3+\mu([x_3,x_1]_\g)v_2+\mu([x_2,x_3]_\g)v_1\\
  &&-\mu(x_3)[v_1,v_2]_\h-\mu(x_2)[v_3,v_1]_\h-\mu(x_1)[v_2,v_3]_\h;\\
  &&\partial_2 (\hat{\rho}+\hat{\omega}_1)(e_1,e_2,e_3)\\
  &=&-\{v_1,\omega_1(x_2,x_3)\}_\h-\{v_2,\omega_1(x_3,x_1)\}_\h-\{v_3,\omega_1(x_1,x_2)\}_\h+\omega_1(\{x_1,x_2\}_\g,x_3)+\omega_1(\{x_2,x_3\}_\g,x_1)\\
  &&+\omega_1(\{x_3,x_1\}_\g,x_2)-\{v_1,\rho(x_2)v_3\}_\h-\{v_2,\rho(x_3)v_1\}-\{v_3,\rho(x_1)v_2\}_\h+\{v_1,\rho(x_3)v_2\}_\h\\
  &&+\{v_2,\rho(x_1)v_3\}_\h+\{v_3,\rho(x_2)v_1\}_\h+\rho(\{x_1,x_2\}_\g)v_3+\rho(\{x_3,x_1\}_\g)v_2+\rho(\{x_2,x_3\}_\g)v_1\\
  &&-\rho(x_3)\{v_1,v_2\}_\h-\rho(x_2)\{v_3,v_1\}_\h-\rho(x_1)\{v_2,v_3\}_\h;\\
  &&{[\hat{\rho}+\hat{\omega}_1,\hat{\mu}+\hat{\omega}_2]}_{\NR}(e_1,e_2,e_3)\\
  &=&-\rho(x_3)\omega_2(x_1,x_2)-\rho(x_3)\mu(x_1)v_2+\rho(x_3)\mu(x_2)v_1+\rho(x_2)\omega_2(x_1,x_3)+\rho(x_2)\mu(x_1)v_3\\
  &&-\rho(x_2)\mu(x_3)v_1-\rho(x_1)\omega_2(x_2,x_3)-\rho(x_1)\mu(x_2)v_3+\rho(x_1)\mu(x_3)v_2-\mu(x_3)\omega_1(x_1,x_2)\\
 &&-\mu(x_3)\rho(x_1)v_2+\mu(x_3)\rho(x_2)v_1+\mu(x_2)\omega_1(x_1,x_3)+\mu(x_2)\rho(x_1)v_3-\mu(x_2)\rho(x_3)v_1\\
 &&-\mu(x_1)\omega_1(x_2,x_3)-\mu(x_1)\rho(x_2)v_3+\mu(x_1)\rho(x_3)v_2.
\end{eqnarray*}
Then
\begin{eqnarray*}
\Big(\partial_1 (\hat{\mu}+\hat{\omega}_2)+\partial_2 (\hat{\rho}+\hat{\omega}_1) +{[\hat{\rho}+\hat{\omega}_1,\hat{\mu}+\hat{\omega}_2]}_{\NR}\Big)((x_1,0),(x_2,0),(0,v_3))=0
\end{eqnarray*}
implies that \eqref{eq:EXT-5} holds.
\begin{eqnarray*}
\Big(\partial_1 (\hat{\mu}+\hat{\omega}_2)+\partial_2 (\hat{\rho}+\hat{\omega}_1) +{[\hat{\rho}+\hat{\omega}_1,\hat{\mu}+\hat{\omega}_2]}_{\NR}\Big)((x_1,0),(0,v_2),(0,v_3))=0
\end{eqnarray*}
implies that \eqref{eq:EXT-6} holds.
\begin{eqnarray*}
\Big(\partial_1 (\hat{\mu}+\hat{\omega}_2)+\partial_2 (\hat{\rho}+\hat{\omega}_1) +{[\hat{\rho}+\hat{\omega}_1,\hat{\mu}+\hat{\omega}_2]}_{\NR}\Big)((x_1,0),(x_2,0),(x_3,0))=0
\end{eqnarray*}
implies that \eqref{eq:EXT-9} holds.

It is  straightforward to see that if
\eqref{eq:EXT-5}, \eqref{eq:EXT-6} and \eqref{eq:EXT-9} hold, then
$$\Big(\partial_1 (\hat{\mu}+\hat{\omega}_2)+\partial_2 (\hat{\rho}+\hat{\omega}_1) +{[\hat{\rho}+\hat{\omega}_1,\hat{\mu}+\hat{\omega}_2]}_{\NR}\Big)(e_1,e_2,e_3)=0.$$

By Proposition \ref{pro:MC-equivalent 2-para}, the conclusion follows.
\end{proof}

Let $(L,[-,-],\partial)$ be a differential graded Lie algebra with $L_0$ abelian. Recall that two Maurer-Cartan elements $P$ and $P'$  are called {\bf gauge equivalent} if  there exists an element
$\xi\in L_0$ such that
\begin{equation}
  P'=e^{\ad_\xi}P-\frac{e^{\ad_\xi} -1}{\ad_\xi}\partial \xi.
\end{equation}

\begin{defi}
Let $(L,[-,-],\partial_1,\partial_2)$ be a bidifferential graded
Lie algebra with $L_0$ abelian. Two Maurer-Cartan elements $(P_1,P_2)$ and
$(P'_1,P'_2)$ of
$(L,[-,-],\partial_1,\partial_2)$ are called {\bf gauge
equivalent} if  there exists an element $\xi\in L_0$ such that
\begin{eqnarray}
  P_1'&=&e^{\ad_\xi}P_1-\frac{e^{\ad_\xi} -1}{\ad_\xi}\partial_1 \xi,\\
  P_2'&=&e^{\ad_\xi}P_2-\frac{e^{\ad_\xi} -1}{\ad_\xi}\partial_2 \xi.
\end{eqnarray}
\end{defi}

It is straightforward to check that
\begin{pro}
 Two Maurer-Cartan elements $(P_1,P_2)$ and $(P_1',P_2')$ of the bidifferential graded Lie algebra $(L,[-,-],\partial_1,\partial_2)$ are  gauge equivalent if and only if for any $k_1,k_2\in\K$, the Maurer-Cartan elements $k_1 P_1+k_2 P_2$ and $k_1 P'_1+k_2 P'_2$ of the differential graded Lie algebra $(L,[-,-],\partial_{k_1,k_2}=k_1\partial_1+k_2\partial_2)$ are gauge equivalent.
\end{pro}

\begin{lem}\label{lem:gauge equivalent}
Let $\g$ and $\h$ be two compatible Lie algebras. Then the two Maurer-Cartan elements $(\hat{\rho}+\hat{\omega}_1,\hat{\mu}+\hat{\omega}_2)$ and $(\hat{\rho'}+\hat{\omega'}_1,\hat{\mu'}+\hat{\omega'}_2)$ of the bidifferential graded Lie algebra $(L=\oplus_k L_k,[-,-]_{\NR},\partial_1,\partial_2)$ are gauge equivalent if and only if there exists an element $\xi\in L_0$ such that the following equations holds:
\begin{eqnarray}
\label{eq:guage equivalent1}   \hat{\rho'}+\hat{\omega'}_1&=&\hat{\rho}+\hat{\omega}_1+[\xi,\hat{\rho}]_{\NR}-\partial_1\xi-\half[\xi,\partial_1\xi]_{\NR},\\
 \label{eq:guage equivalent2}  \hat{\mu'}+\hat{\omega'}_2&=&\hat{\mu}+\hat{\omega}_2+[\xi,\hat{\mu}]_{\NR}-\partial_2\xi- \half[\xi,\partial_2\xi]_{\NR}.
\end{eqnarray}
\end{lem}
\begin{proof}
By Proposition \ref{pro:MC-equivalent 2-para}, the two Maurer-Cartan elements $(\hat{\rho}+\hat{\omega}_1,\hat{\mu}+\hat{\omega}_2)$ and $(\hat{\rho'}+\hat{\omega'}_1,\hat{\mu'}+\hat{\omega'}_2)$ are gauge equivalent if and only if there exists an element
$\xi\in L_0$ such that
\begin{eqnarray}
\label{eq:guage equivalent3} \hat{\rho'}+\hat{\omega'}_1&=&e^{\ad_\xi}(\hat{\rho}+\hat{\omega}_1)-\frac{e^{\ad_\xi} -1}{\ad_\xi}\partial_1 \xi,\\
\label{eq:guage equivalent4} \hat{\mu'}+\hat{\omega'}_2&=&e^{\AD_\xi}(\hat{\mu}+\hat{\omega}_2)-\frac{e^{\AD_\xi} -1}{\AD_\xi}\partial_2 \xi.
\end{eqnarray}
By the definition of bidegrees, we have $||\xi||=1|-1$,  $||\hat{\rho}||=||\hat{\mu}||=1|0$ and $||\hat{\omega}_1||=||\hat{\omega}_2||=2|-1$. Then we have $$\ad_\xi^n(\hat{\rho})=\AD_\xi^n(\hat{\mu})=0,\quad n\geq 2,\quad \ad_\xi^n(\hat{\omega}_1)=\AD_\xi^n(\hat{\omega}_2)=0,\quad n\geq 1.$$
Thus we have
\begin{eqnarray*}
e^{\ad_\xi}(\hat{\rho}+\hat{\omega}_1)&=&\hat{\rho}+\hat{\omega}_1+[\xi,\hat{\rho}]_{\NR},\\
e^{\AD_\xi}(\hat{\mu}+\hat{\omega}_2)&=&\hat{\mu}+\hat{\omega}_2+[\xi,\hat{\mu}]_{\NR}.
\end{eqnarray*}
Similarly, we have
\begin{eqnarray*}
\frac{e^{\ad_\xi} -1}{\ad_\xi}\partial_1 \xi&=&\partial_1\xi+\half[\xi,\partial_1\xi]_{\NR},\\
\frac{e^{\AD_\xi} -1}{\AD_\xi}\partial_2 \xi&=&\partial_2\xi+ \half[\xi,\partial_2\xi]_{\NR}.
\end{eqnarray*}
Thus \eqref{eq:guage equivalent3} and \eqref{eq:guage equivalent4}
are equivalent to \eqref{eq:guage equivalent1} and \eqref{eq:guage
equivalent2}, respectively.
\end{proof}

Now we are ready to give the main result in the section which classify nonabelian extensions of compatible Lie algebras.

\begin{thm}
 Let $\g$ and $\h$ be two compatible Lie algebras. Then there is a one-to-one correspondence between the isomorphism classes of nonabelian
extensions of the compatible Lie algebra $\g$ by $\h$ and the gauge equivalence classes of Maurer-Cartan elements in the bidifferential graded Lie algebra $(L=\oplus_k L_k,[-,-]_{\NR},\partial_1,\partial_2)$.
\end{thm}
\begin{proof}
  By a direct calculation, for all $x,y\in\g$ and $u,v\in \h$, we have
  \begin{eqnarray*}
    &&\big(\hat{\rho'}+\hat{\omega'}_1-\hat{\rho}-\hat{\omega}_1-[\xi,\hat{\rho}]_{\NR}+\partial_1\xi+\half[\xi,\partial_1\xi]_{\NR}\big)((x,u),(y,v))\\
    &=&\rho'(x)v-\rho(x)v-[\xi(x),v]_\h-\rho'(y)u+\rho(y)u+[\xi(y),u]_\h\\
    &&+\omega_1'(x,y)-\omega_1(x,y)-\rho(x)\xi(y)+\rho(y)\xi(x)+\xi([x,y]_\g)-[\xi(x),\xi(y)]_\h,
  \end{eqnarray*}
  which implies that \eqref{eq:guage equivalent1} holds if and only if \eqref{eq:isomorphism1} and \eqref{eq:isomorphism3} hold. Similarly, \eqref{eq:guage equivalent2} holds if and only if \eqref{eq:isomorphism2} and \eqref{eq:isomorphism4} hold. By Lemma \ref{lem:gauge equivalent} and Proposition \ref{pro:JMC}, the conclusion follows.
\end{proof}

 \end{document}